\documentclass[conference]{IEEEtran}
\usepackage{times}

\usepackage[numbers]{natbib}
\usepackage{multicol}
\usepackage[bookmarks=true]{hyperref}

\usepackage{graphicx}
\usepackage{amsthm} 
\usepackage{amsmath} 
\usepackage{amssymb}  
\usepackage{url}
\usepackage[font=footnotesize]{caption}
\usepackage[font=footnotesize,format=hang]{subcaption}
\usepackage{color}
\usepackage{verbatim}
\usepackage[multiple,bottom,hang,flushmargin]{footmisc}
\usepackage{algorithm}
\usepackage{algorithmic}
\usepackage{booktabs}
\usepackage{multirow}
\usepackage{float}
\usepackage[table]{xcolor}
\usepackage{doi}

\newcommand{\R}{\mathbb{R}}

\newtheorem{theorem}{Theorem}
\newtheorem{lemma}[theorem]{Lemma}
\newtheorem{assumption}{Assumption}

\begin{document}

\title{Convex Computation of the Basin of Stability to Measure the Likelihood of Falling: A Case Study on the Sit-to-Stand Task}

\author{Victor Shia, Talia Moore, Ruzena Bajcsy, Ram Vasudevan}

\author{
    \IEEEauthorblockN{Victor Shia\IEEEauthorrefmark{1}, Talia Moore\IEEEauthorrefmark{2}, Ruzena Bajcsy\IEEEauthorrefmark{1}, Ram Vasudevan\IEEEauthorrefmark{3}}
    \IEEEauthorblockA{\IEEEauthorrefmark{1}Electrical Engineering and Computer Sciences \\
    University of California, Berkeley \\
    \{vshia, bajcsy\}@berkeley.edu}
    \IEEEauthorblockA{\IEEEauthorrefmark{2}
    Organismic and Evolutionary Biology \\
    Harvard University \\
    taliaym@gmail.com}
    \IEEEauthorblockA{\IEEEauthorrefmark{2}Mechanical Engineering \\
    University of Michigan
    \\ramv@umich.edu}
}

\maketitle
\begin{abstract}
Locomotion in the real world involves unexpected perturbations, and therefore requires strategies to maintain stability to successfully execute desired behaviours. 
Ensuring the safety of locomoting systems therefore necessitates a quantitative metric for stability. 
Due to the difficulty of determining the set of perturbations that induce failure, researchers have used a variety of features as a proxy to describe stability.  
This paper utilises recent advances in dynamical systems theory to develop a personalised, automated framework to compute the set of perturbations from which a system can avoid failure, which is known as the basin of stability.
The approach tracks human motion to synthesise a control input that is analysed to measure the basin of stability.
The utility of this analysis is verified on a Sit-to-Stand task performed by $15$ individuals. 
The experiment illustrates that the computed basin of stability for each individual can successfully differentiate between less and more stable Sit-to-Stand strategies.
\end{abstract}



\section{Introduction}\label{sec:introduction}

Falls are a leading cause of accidental injury and death throughout much of the world.
Due to the aging population and the outsized impact falling has on the elderly, the cost associated with falls is expected to rise dramatically in the next twenty years \cite{CDC2015}.
Directed therapeutic care can significantly reduce the risk of falling~\cite{Robertson2001,Horak2006}; however, the resources available for such treatment are limited.
An automated test identifying individuals at risk for falling can make targeted deployment of therapeutic care feasible.
Unfortunately the construction of such a test has been challenging.

This paper develops a personalised automated diagnostic test that uses kinematic observations to measure an individual's likelihood of falling. The approach, which is grounded in dynamical systems theory, computes the \emph{Basin of Stability} (BOS) of a locomotor pattern, or the set of perturbations that do not lead to a fall under an individual's chosen locomotor strategy (illustrated in Figure \ref{fig:example_bos}).
Informally, an individual that is able to tolerate a larger set of perturbations has a larger BOS and is less likely to fall. 

In fact, measuring the BOS is a direct way to characterise the likelihood of falling, since it identifies the specific deficiencies that lead to failure~\cite{Pollock2000}. 
Unfortunately the computation of this individual- and behaviour-specific BOS is challenging, since it requires measuring the effect of arbitrary perturbations to a nonlinear system. 
An empirical experiment would require exhaustive perturbation of an individual throughout a locomotor pattern, which is practically infeasible and dangerous.

To address these issues, the presented approach computes the BOS in a tractable manner using convex optimization.
Though the method is applicable to arbitrary locomotor patterns, this paper illustrates the utility of this technique by analysing Sit-to-Stand (STS) manoeuvres, 
STS manoeuvres are less complex than other locomotor patterns (e.g. walking, running, climbing, lifting), simplifying the validation of the method.
Although comparatively straightforward, the ability to stand is a prerequisite for bathing, cooking, dressing, maintaining hygiene, and walking.  
As a result, difficulty in performing STS manoeuvres is considered a primary risk factor for falls amongst the elderly~\cite{Campbell1989}.

\begin{figure}
  \begin{subfigure}[b]{0.5\textwidth}
    \centering
    \includegraphics[width=\textwidth,clip,trim=1cm 8.5cm 1.5cm 1cm]{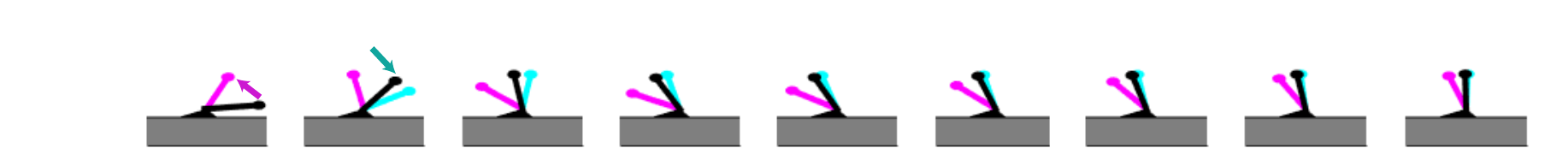}
    \caption{Pictorial Model}
    \label{fig:example_bos_a}
  \end{subfigure}
  \centering
  \begin{subfigure}[b]{0.5\textwidth}
    \centering
    \includegraphics[width=\textwidth,clip,trim=1cm 1cm 1.5cm 3cm]{images/example_ipm}
    \caption{Basin of Stability over Time}
    \label{fig:example_bos_b}
  \end{subfigure}
  \caption{The BOS of an inverted pendulum swing-up controller with perturbations indicated by arrows.  (\subref{fig:example_bos_a}) illustrates a nominal (black) and perturbed (magenta and cyan) swing-up trajectory for the inverted pendulum. (\subref{fig:example_bos_b}) illustrates the BOS (gray) of the nominal controller, the nominal trajectory (black), and the perturbed trajectory (magenta and cyan) in the configuration space of the inverted pendulum\protect\footnotemark. Despite perturbation, the magenta and cyan trajectories arrive at the upright configuration because they remain within the basin of stability of the nominal trajectory. 
    }
  \label{fig:example_bos}
\end{figure}

\footnotetext{Note that at time $T$, the BOS of the controller is small because we are concerned about a finite-time BOS.}

\subsection{Existing Stability Metrics}

Due to the importance of STS manoeuvres in maintaining quality of life, and the impossibility of testing all possible perturbations, a variety of methods to characterise an individual's likelihood of falling while performing STS have been proposed.  These methods are summarised in Table \ref{tab:sts_metrics}.

\begin{table*}
  \begin{tabular}{|c | p{11.5cm} | p{1.5cm} |}
  \toprule
    Methods & Summary & References \\ \hline
    BERG Balance Test & A battery of functional tests with a single number that determines the likelihood of falling & \cite{Berg1989} \\ \hline
    Stops Walking When Talking & Relates the amount of attention a person requires to perform an action with a likelihood of falling & \cite{Lundin-Olsson1997} \\ \hline
    Timed Up \& Go & Correlates a likelihood of falling with the the amount of time it takes to stand up & \cite{Podsiadlo1991,Lundin-Olsson1998} \\ \hline
    Model based methods & Uses a single inverted pendulum to determine the set of feasible initial positions or positions and velocities that can lead to standing up &  \cite{Pai1997,Papa1999,Patton1999,Papa2000,Pai2003,Fujimoto2013} \\
    \bottomrule
  \end{tabular}
  \caption{Various STS Stability Methods}
  \label{tab:sts_metrics}
\end{table*}

These methods generally summarise STS motions using a single feature and perform versions of regression analysis to estimate a patient's stability.
In doing so, they forfeit the ability to characterise the specific deficiencies limiting an individual's STS ability.
More troublingly, according to several studies, the ability of these clinical tests to distinguish between stable and unstable patients is unclear~\cite{Lundin-Olsson1997,BogleThorbahn1996,Steffen2002,Muir2008a,Shumway-Cook2000}. 

\subsection{Dynamical Systems Perspective}

Computing the BOS is a fundamental objective of the dynamical systems community since it can be used to verify the satisfactory operation of a system despite perturbations. 
In particular, engineers have long sought to understand the behavior of a dynamical system after arbitrary perturbation. 
Although it is possible to directly simulate arbitrary configurations (analogous to exhaustive perturbation), this technique provides only limited insight. 
To address this issue, the dynamical systems community has studied numerical methods to compute the BOS.

These methods for nonlinear systems include Lyapunov-based techniques~\cite{Prajna2004b,Topcu2007} and Hamilton-Jacobi based methods~\cite{Mitchell2005}.
Lyapunov-based methods~\cite{Lyapunov1992} search for functions whose sub-level sets satisfy certain criteria. 
The construction of such a Lyapunov function is possible for polynomial dynamical systems using semidefinite programming~\cite{Parrilo2000}, but requires solving a challenging bilinear optimization problem, limiting potential applicability.
Hamilton-Jacobi based methods discretise the domain and run variants of dynamic programming on a discretised nonlinear partial differential equation to determine the set of states that belong to a BOS.
Though this method is able to tractably compute the BOS for dynamical systems with special structure~\cite{maidens2013lagrangian}, it is only able to accurately compute the BOS for general systems with less than $4$ states. 
Recently, the authors developed a method to analytically compute the BOS for polynomial dynamical systems based on occupation measures. This method, which relies on convex optimization and is described in further detail below, tractably outer approximates the BOS of a system without relying upon exhaustive perturbative experiments or simulation. Furthermore, this method successfully synthesises safe robotic motion for systems with up to $8$ states~\cite{Henrion2013,Majumdar2014,Shia2014b,Mohan2016}.  

\subsection{Contribution}

First, Section \ref{sec:method} describes a personalised computational framework to model, identify, and analyse the unique stability of an individual's motion. 
Second, Sections \ref{sec:experiments} and \ref{sec:results} describe a motion capture dataset of humans performing various STS strategies. The proposed methods then evaluate each individual's kinematic stability.
Section \ref{sec:discussion} summarises the impact of the proposed method and describes potential extensions.
\section{Methodology} \label{sec:method}

This section presents the framework to compute the BOS of an individual's locomotor pattern given kinematic data. 
The approach is summarised informally in Algorithm \ref{alg:compute_bos}. 
The steps are described abstractly in this section to ensure straightforward generalization to arbitrary locomotor patterns. 
In Sections \ref{sec:experiments} and \ref{sec:results}, a concrete instantiation of each step is described in the case of STS motion.

\begin{figure}[t]
\begin{minipage}{0.48\textwidth}
	\begin{algorithm}[H]
		  \caption{Computing the Reachable Set for STS}
		  \label{alg:compute_bos}
		  \begin{algorithmic}[1]
		    \STATE Given: observations $x_\text{obs}$ of motion
		    \STATE Choose a model for the STS motion (Section \ref{subsec:model}). 
		    \STATE Run optimal control to find $u_\text{obs}$ (Section \ref{subsec:nom}).
		    \STATE Construct controller to track $x_\text{obs}$ (Section \ref{sec:control_design}).
		    \STATE Compute the backwards reachable set (Section \ref{sec:om}).
		  \end{algorithmic}
	\end{algorithm}
\end{minipage}
\end{figure}

\subsection{Preliminaries}
The notation used throughout the remainder of this paper is presented in this section. 
Let $\mathbb{R}^n$ be a n-dimensional set of real numbers.
Let $X \subset \R^n$ be a compact set.
Let $[0,T]$ denote a time interval of interest.
Let $C^1(X,\R)$ be the space of continuously differentiable functions from $X$ to $\R$.
Let $L^2(X,\R)$ be the space of square integrable functions under the Lebesgue measure from $X$ to $\R$.
Let $\mathbb{R}_n[x]$ be the set of polynomials in $x$ with maximum total degree $n$.

\subsection{Model and Observations}
\label{subsec:model}

Next, suppose that the dynamical model describing the motion of an individual is:
\begin{align}
  \dot{x}(t) &= f_\phi(t,x) + g_\phi(t,x)u(t,x) \nonumber \\
  x &\in [\underline{x}, \overline{x}] \subset \R^n  \label{eq:cas}\\
  u &\in [\underline{u}, \overline{u}] \subset \R^m \nonumber
\end{align}

\noindent where $X = [\underline{x},\overline{x}] \subset \R^n$ represents the state space of the model, $f \in C^1([0,T] \times \R^n,\R^n) $ and $g \in C^1([0,T] \times \R^n,\R^m) $ describe how the input $u \in L^2([0,T] \times \R^n, \R^m)$ affect the dynamics, $\phi$ represents the individual specific parameters of the model (e.g. mass, limb length, moment of inertia, etc.), and $\underline{u},\overline{u} \in\R^m$ represent input bounds. 
As each individual is different, $\phi$, $\underline{x}$, $\overline{x}$, $\underline{u}$ and $\overline{u}$ are distinct for each individual and must be identified as described in further detail in Section \ref{sec:experiments}. 
This paper assumes that direct observations of the state trajectory of a nominal locomotor pattern, $x_\text{obs} \in C^1([0,T],\R^n)$, are available. 
This can be constructed after interpolation from a variety of data sources as described in Section \ref{sec:experiments}.
 
\subsection{Identifying an Input from Observations} 
\label{subsec:nom}

After selecting a model, the input, $u_{\text{obs}}:[0,T] \to \R^m$ that generates the given observations must be constructed.
There are two methods for determining the input for the observed motion: inverse dynamics and optimal control.
Inverse dynamics uses the observed variables $x_{\text{obs}}:[0,T] \to X$ to estimate $\dot{x}_{\text{obs}}:[0,T] \to \R^n$.
$u_\text{obs}$ can then be computed for all $t$ in Equation \eqref{eq:cas} using a known $(x_{\text{obs}},\dot{x}_{\text{obs}})$~\cite{Koopman1995}.
As the inverse kinematic solution is sensitive to noise in $x_\text{obs}$, optimal control is used in this paper to compute $u_\text{obs}$.
Optimal control instead calculates $u_\text{obs}$ via the optimization problem: 

\begin{align}
  \underset{u_\text{obs} \in L^2([0,T],\R^m)} {\inf}  \hspace{0.5cm} & \int_0^T \| x(t) - x_\text{obs}(t) \|_2^2 dt  \label{eq:opt_control} \\
  \text {s.t.} \hspace{0.25cm} \dot{x}(t) &= f_\phi(t, x) + g_\phi(t, x) u_\text{obs}(t) \; &\forall t \in [0,T] \nonumber \\
  \phantom{\text{s.t.} \hspace{0.25cm}} x(t) &\in X \; &\forall t \in [0,T]\nonumber \\
  \phantom{\text{s.t.} \hspace{0.25cm}} u_{obs}(t) &\in [\underline{u}, \overline{u}] \; &\forall t \in [0,T] \nonumber
\end{align}

The solution to this problem is a feedforward open loop control input $u_\text{obs}$ that minimises the $L^2$ error between the state trajectory and the observed trajectory.
After treating the nominal input $u_\text{obs}$ as a polynomial function, collocation~\cite{Hargraves1987} is used to transform this optimal control problem into a nonlinear optimization program, which can be efficiently solved by a variety of nonlinear programming solvers. 

\subsection{Feedback Controller Design}\label{sec:control_design}

Neuroscientists, psychologists, motor control researchers, and biomechanists have observed that the nominal trajectories humans follow during locomotor patterns are robust to small perturbations~\cite{Flash1985,Uno1989,Kawato1999,Todorov2002,Cusumano2013}.
This robustness is conferred by feedback about the nominal control input or goal.
To date, experimental research has been unable to identify an overall strategy that endows such robustness. 
 
For example, research has shown that subjects minimise the square of jerk during reaching tasks~\cite{Flash1985}.
Alternatively, others have shown that for endpoint reaching tasks, subjects utilise a time-varying Proportional Derivative (PD) control to reach a specified endpoint~\cite{Liu2007}.  
For the lower body, other researchers tracked the evolution of step width and found that subjects tended to correct deviations with just a proportional controller~\cite{Dingwell2010,Dingwell2013}. 

To imbue the feedforward nominal control input that is identified by the optimal control algorithm in Section \ref{subsec:nom} with this feedback robustness, the following assumptions are made:

\begin{assumption}
For each distinct locomotion action, humans utilise a feedforward control law with corresponding feedback.  
To perform a different action, the subject switches control laws.
\label{as:nom_traj}
\end{assumption}
\noindent Assumption \ref{as:nom_traj} states that for a particular action, such as standing slowly, the subject follows a combination of feedforward and feedback control laws.
If a specific control law is not able to take a subject to standing after perturbation, a subject must switch control laws to stand safely. 

\begin{assumption}
While performing a specific action, the subject utilises a PD feedback around a nominal trajectory, $x_\text{obs}$ to correct deviations in the trajectory.
\label{as:pd_control}
\end{assumption}
\noindent According to Assumption \ref{as:pd_control} the feedback control law is:
  \begin{align}
    u(t,x) &= u_\text{obs}(t) + u_{cc}(t,x) \label{eq:pid_controller}\\
    &= u_\text{obs}(t) + K (x(t) - x_\text{obs}(t)) \nonumber
  \end{align}

where $u_{cc}$ represents the general form of the feedback controller and $K$ represents the PD controller gain acting on the states and observation.
Note, the method presented to estimate the BOS (described in Section \ref{sec:om}) can handle more general nonlinear feedback control inputs. However, as described earlier, the existing literature suggests that humans apply only linear feedback~\cite{Dingwell2010,Dingwell2013}. 

If the gain $K$ on the feedback controller is selected too rigidly, then the control law will oscillate around the desired trajectory rather than converging to the final state of the desired trajectory. 
This can be avoided with sufficiently small gains, as illustrated by the system in Figure \ref{fig:example_bos}.  
To determine this satisfactory feedback gain $K$, we apply a Linear Quadratic Regulator (LQR) algorithm to determine the optimal state feedback law $u$ that minimises a quadratic cost: 
\begin{align}
    &\underset{u_{cc} \in L^2([0,T],\R^m)}{\min} && \hspace{-0.25cm}\int_0^T \Big( (x(t)-x_\text{obs}(t))^T Q (x(t)-x_\text{obs}(t)) \nonumber  
  \label{eq:pid_control} \\
    & & & \hspace{2cm}  +  u_{cc}(t)^T R u_{cc}(t) \Big) dt \\
    & \text{s.t.} && \hspace*{-1.5cm} \dot{x}(t) = Ax(t) + B(u_\text{obs}(t)+u_{cc}(t,x)) \; \forall t \in [0,T] \nonumber
  \end{align}

By selecting $Q=\frac{I}{2}$ and $R=0.005I$ where $I$ is the identity matrix of appropriate dimension, the resulting controller is designed to minimise the $Q$-weighted $L^2$ error of $x$ from $x_\text{obs}$.
For linear systems, the LQR problem has a closed form solution provided by the Algebraic Ricatti Equation described by a linear state feedback law~\cite{Callier1994}. 
For the purposes of this paper, small-angle approximations are used to linearise $f_\phi$ and $g_\phi$ to obtain $A$ and $B$. 

\begin{assumption}
The torque limits are constant throughout the motion.
\label{as:torque_bounds}
\end{assumption}
\noindent As humans do not have the ability to apply arbitrary torque to any joint, individual-specific torque limits $[\underline{u}, \overline{u}]$ are set to the minimum and maximum of $u_\text{obs}$ generated from the optimal control. 

\subsection{Computing the Basin of Stability} \label{sec:om}
Given a model, input bounds, and feedback control input that tracks a nominal observation, the BOS can be formally defined as follows: \textit{the BOS is the set of states as a function of time that can be driven by the feedback control input to a target configuration, $X_T \subset X$, by time $T$}.
In the case of STS, the target set $X_T$ corresponds to the set of states where the subject is standing. 
For brevity, a modified optimization algorithm inspired by \cite{Henrion2013} to compute this BOS is presented:
\begin{subequations} 
\begin{align} 
		& \underset{v \in C^1([0,T] \times \R^n,\R^n)}{\text{inf}} & & \hspace*{-0.75cm} \int\limits_{[0,T] \times X} v(t,x) dtdx && (D) \nonumber \\
		& \text{s.t.} & & \hspace*{-2cm} \frac{\partial v(t,x)}{\partial x} \left(f_\phi(t,x)+g_\phi(t,x)u(t,x)\right)  \label{eq:opt_lyap} \\
		& & & \hspace*{-0.5cm} +  \frac{\partial v(t,x)}{\partial t} \leq 0 && \hspace*{-1.5cm} \forall (t,x) \in [0,T] \times X \nonumber \\
		& & & \hspace*{-0.75cm} v(t,x) \geq 0 && \hspace{-1.5cm} \forall (t,x) \in [0,T] \times X  \\
		& & & \hspace*{-0.75cm} v(T,x) \geq \alpha && \hspace{-1.5cm} \forall x \in X_T \label{eq:opt_xt}
\end{align}
\end{subequations} 

\noindent where $\alpha > 0$ is a parameter that can be selected by the user.
To understand the relationship between the solution to this optimization problem ($v$) and the BOS, notice that $v(t,x) \geq \alpha$ for points that belong on the BOS:
\begin{lemma}
If $v$ is a solution to ($D$), then $v(t,\cdot) \geq \alpha$ on the BOS.
\end{lemma}
\begin{proof}
Suppose $x:[0,T] \to X$ is a trajectory of the model that reaches $X_T$.
Notice that $x(t)$ is in the BOS for all $t \in [0,T]$.
Select an arbitrary $t_s \in [0,T]$, then:
  \begin{align}
    \alpha &\leq v(T,x(T)) \\
    &= v(t_s, x(t_s)) + \int_{t_s}^T \left(\frac{\partial v(t,x(t))}{\partial x} \left(f_\phi(t,x(t)) \right) \right. \\
    & \hspace{0.5cm}\left.+\frac{\partial v(t,x(t))}{\partial x} \left(g_\phi(t,x(t))u(t,x(t))\right) + \frac{\partial v}{\partial t}(t,x(t))\right) dt \\
    &\leq v(t_s, x(t_s)) 
  \end{align}

since $\frac{\partial v(t,x(t))}{\partial x} \left(f_\phi(t,x(t))+g_\phi(t,x(t))u(t,x(t))\right) + \frac{\partial v}{\partial t}(t,x(t)) \leq 0$ on $[0,T] \times X$ and $v(T,\cdot) \geq \alpha$ on $X_T$.
The desired result follows.
\end{proof}

The intuition of the proof is as follows: if $v(T,\cdot) \geq \alpha$ on $X_T$ \eqref{eq:opt_xt}, since $v$ must decrease as the system evolves \eqref{eq:opt_lyap}, for a point $(t,x) \in [0,T] \times X$ to reach $X_T$, $v(t,x(t)) \geq \alpha$ must hold for all time.
As a result, the $\alpha$ super-level set of $v$ at each time $t$ in $[0,T]$ can be used as a test to determine whether a point does not belong to the BOS of the motion under consideration. 
In Figure \ref{fig:example_bos}, for example, the light gray region denotes the $v(t,x) \geq \alpha$ level set with the dark gray region denoting different time slices of the $v(t,x) \geq \alpha$ level set.
Outside of the points that belong to BOS, the optimization problem tries to minimise $v$ by bringing it as close to $0$ as possible.
Several recent papers formally describe the convergence of this approach, which we do not include here for the sake of brevity~\cite{Henrion2013,Majumdar2014,Shia2014b}.

To solve ($D$) numerically, the dynamics are assumed to be polynomial and the state space and target set are assumed to be semi-algebraic sets. 
Since by the Stone-Weierstrass Theorem polynomial functions are able to approximate the behavior of other continuous functions on a compact domain~\cite{rudin1964principles}, this assumption is made without too much loss in generality.
The positivity constraints are converted to sum-of-squares constraint~\cite{Parrilo2000}.
The result is a semidefinite optimization program that tractably constructs an outer approximation to the BOS~\cite{Shia2014b}. 
\section{Experiments} \label{sec:experiments}

This section describes a formal implementation of the method presented in Section \ref{sec:method} and an experiment constructed to evaluate its validity.
One method to verify the correctness of a computed BOS is via direct perturbative experiments; however, these experiments can be prohibitive and are dangerous.
Instead we utilise observations from motor control research to validate the computed stability estimates of distinct STS maneuvers performed by each subject.

\subsection{Intuition from Motor Control} \label{subsec:motor}

Due to the time delay of the nervous system, motor control researchers have hypothesised that the response of perturbations to fast motions is largely governed by open-loop reflex responses \cite{Full2002,Spagna2007}, whereas slower motions allow a closed-loop correcting response to perturbations.
Based on this experimentally validated tradeoff between speed and feedback~\cite{Maki1997a,Wand1980,Hof2010}, we expect slower movements to have a larger basin of stability.

The open- and closed-loop control laws are exemplified by two distinct STS strategies: momentum-transfer and quasi-static~\cite{Hughes1994,Aissaoui1999}, shown in Figure \ref{fig:example_sts_motions}. 
The momentum-transfer strategy (indicated throughout in orange) consists of swinging one's trunk forward rapidly, using the forward momentum of the upper body to stand up. This strategy requires significant postural control, due to a dynamically unstable transition phase ~\cite{Hughes1994}.
In contrast, the quasi-static strategy (indicated throughout in green) consists of leaning forward while sitting to position the centre of mass (COM) above the feet, then using as little momentum as possible to slowly stand. 
The motion is statically stable at any given moment, but requires more energy to perform than the momentum-transfer strategy~\cite{Anan2012}.
Natural STS movements likely form a continuum between the open-loop momentum transfer and the closed-loop quasi-static strategies. 

To validate Algorithm \ref{alg:compute_bos} experimentally, subjects performed STS using their preferred strategy at two speeds and the momentum transfer and quasi-static strategies.
Computed results are considered accurate if (1) the slower preferred strategy has a larger BOS than the faster preferred strategy and (2) the quasi-static strategy has a larger BOS than the dynamic strategy for the same individual.

\begin{figure}[h]
  \begin{subfigure}[b]{0.5\textwidth}
    \centering
    \includegraphics[width=\textwidth,clip,trim={0cm 0cm 0cm 0cm}]{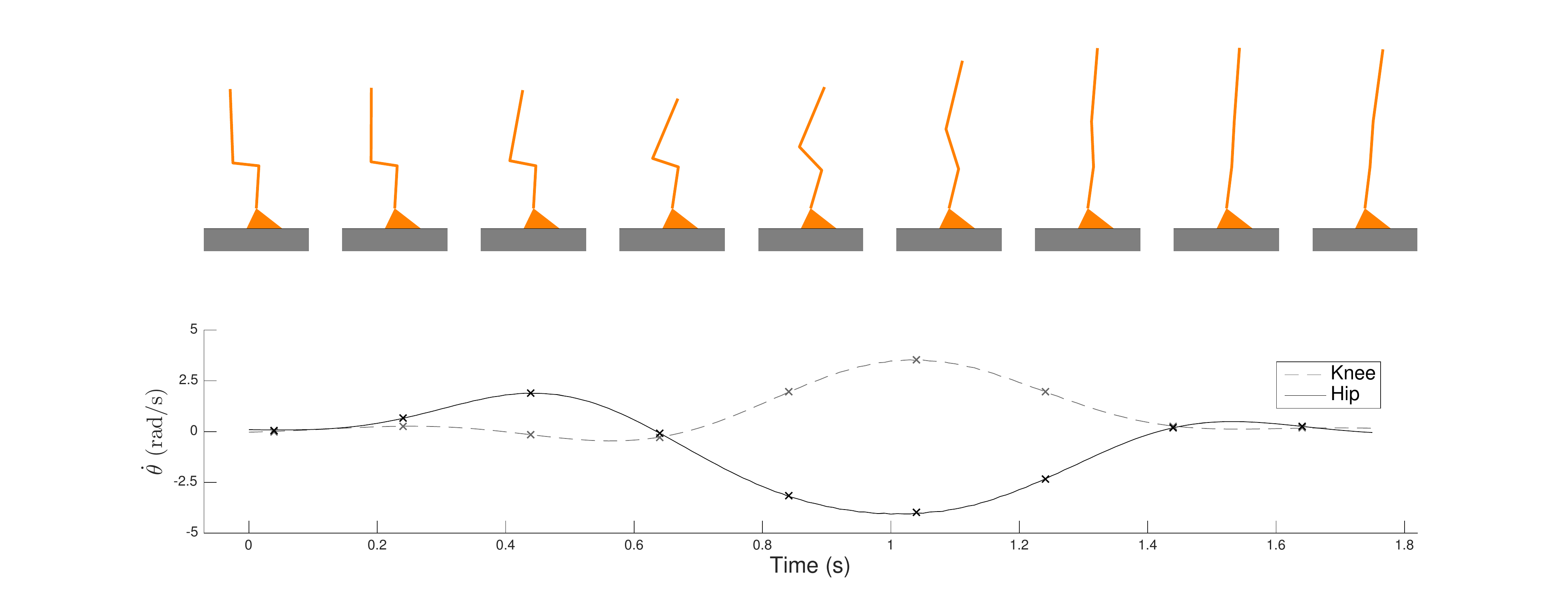}
   \caption{Momentum Transfer STS diagram (top) and velocity profile (bottom).\newline}
   \label{fig:example_mt_sts}
  \end{subfigure}
  \centering
  \begin{subfigure}[b]{0.5\textwidth}
    \centering
    \includegraphics[width=\textwidth,clip,trim={0cm 0cm 0cm 0cm}]{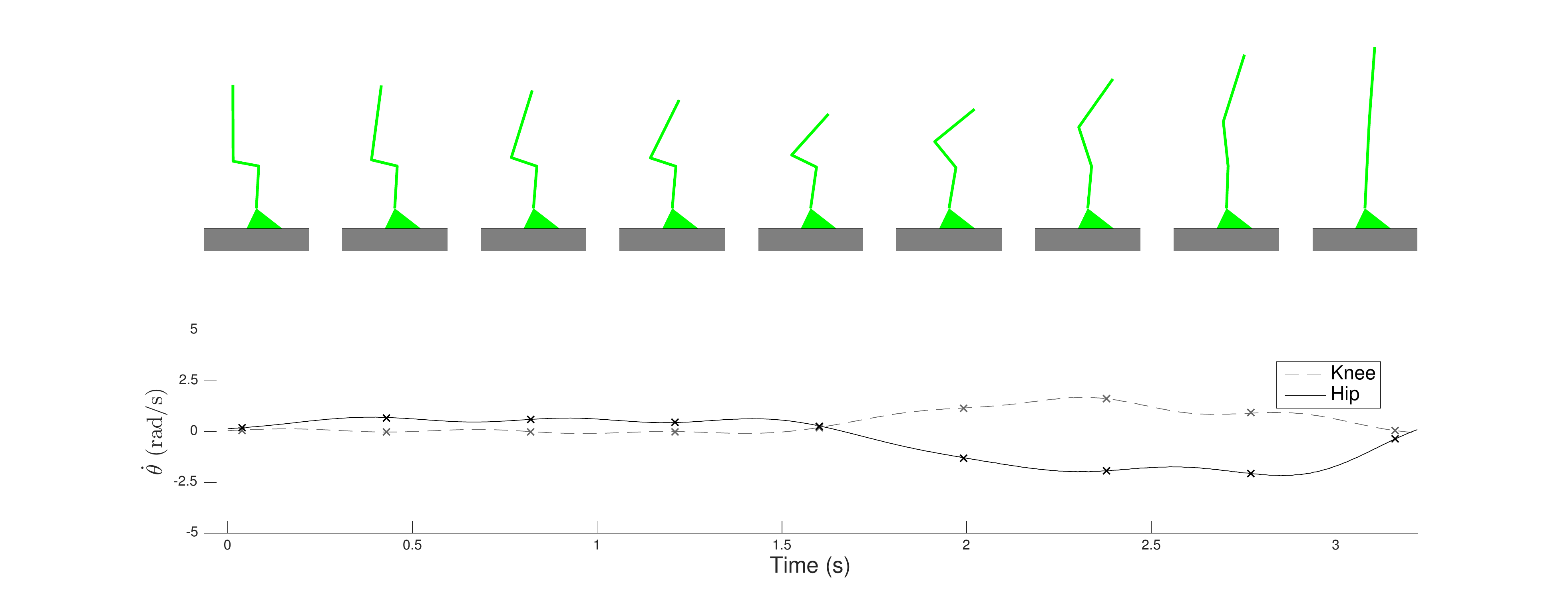}
    \caption{Quasi-static STS diagram (top) and velocity profile (bottom).}
    \label{fig:example_qs_sts}
  \end{subfigure}
  \caption{An illustrate of the two STS Strategies used to perform validation. The difference in angular velocity of the ankle joint is negligible between the two motions and is not shown.}
  \label{fig:example_sts_motions}
\end{figure}

\subsection{Data Collection}\label{subsec:data}

Subjects began in a seated position with their trunk and tibiae oriented vertically, and arms crossed. The chair height was adjusted such that the subject's femurs were parallel to the ground. Subjects wore a customised motion capture suit with 43 PhaseSpace markers (shown in Figure \ref{fig:mocap_suit}). 
STS movements were recorded using an AMTI OPT464508 force plate~\cite{OPT464508} under the subject's feet\footnote{Force measurements were used to determine the start and end time of the STS motion} and a PhaseSpace Impulse X2 motion capture system with 8 infrared cameras~\cite{PhaseSpace} (Figure \ref{fig:exp_setup}).   
Force data were collected at 2400Hz, motion capture data were collected at 480Hz, and the subject's skeleton was extracted using PhaseSpace's Recap2 software~\cite{Recap2}.
Both the motion capture and force plate data were smoothed using a 4th-order Butterworth filter with a cut-off frequency of 2Hz.  

\begin{figure}
  \begin{subfigure}[b]{0.28\textwidth}
    \centering
    \includegraphics[height=6cm,clip,trim={1cm 1cm 1cm 1cm}]{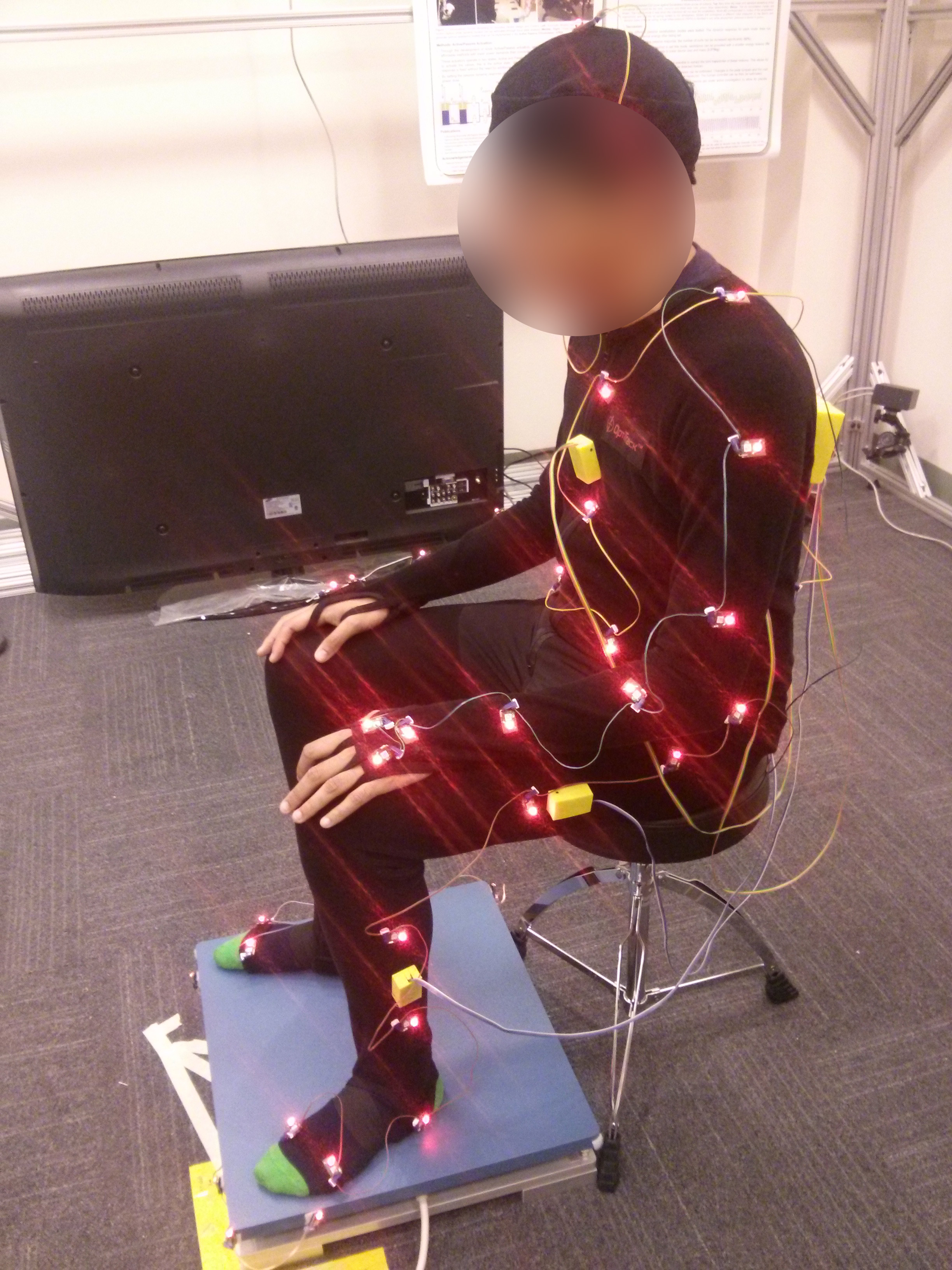}
   \caption{Seated subject prior to STS motion with feet on the force plate.\newline}
   \label{fig:exp_setup}
  \end{subfigure}
  \centering
  \begin{subfigure}[b]{0.2\textwidth}
    \centering
    \includegraphics[height=6cm,clip,trim={0.2cm 0.2cm 0.2cm 0.2cm}]{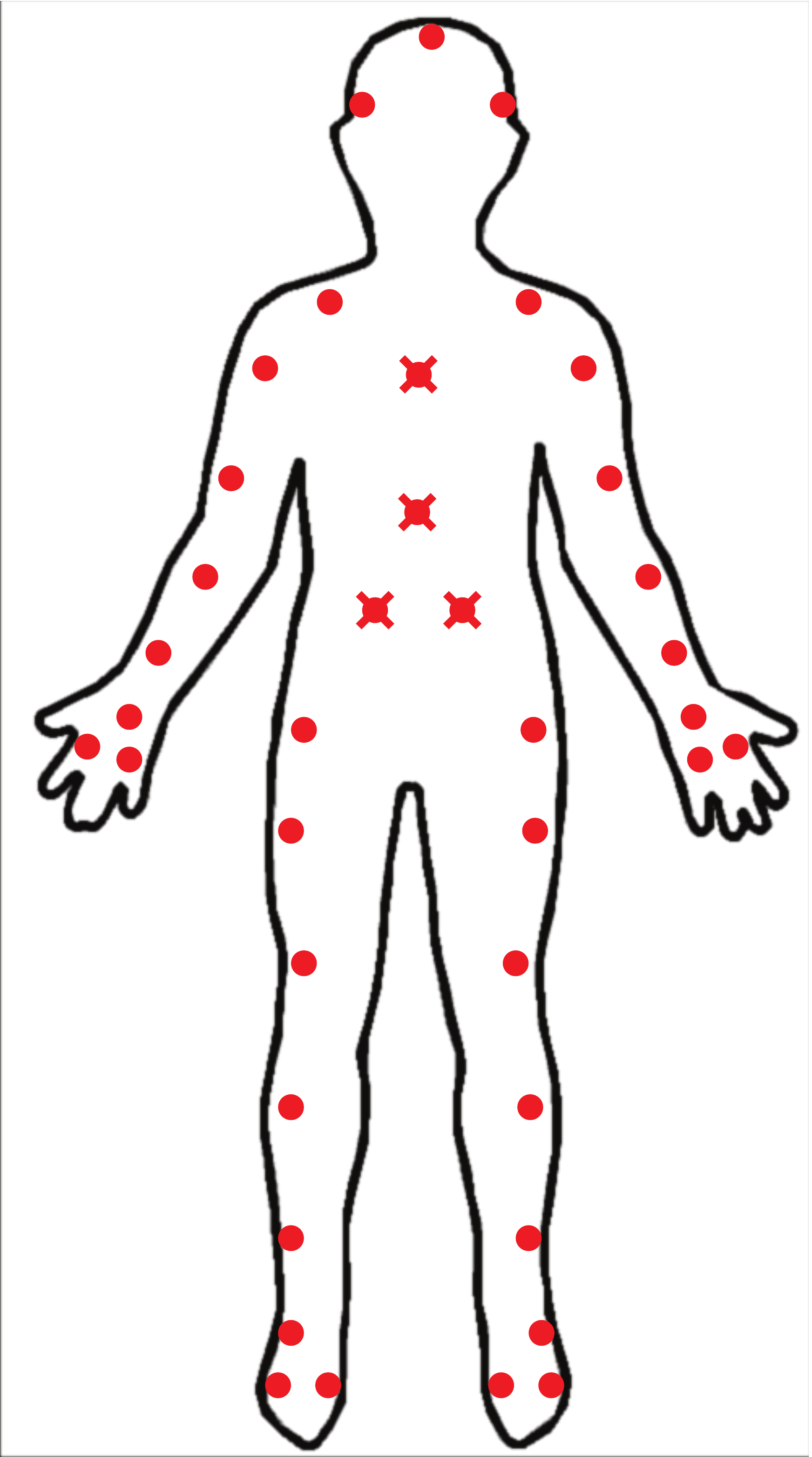}
    \caption{Circles indicate markers placed on front, Xs indicate markers placed on back.}
    \label{fig:mocap_suit}
  \end{subfigure}
  \caption{Experimental Setup with LED Marker Placement}
\end{figure}

We collected data from 2 cohorts: 10 young and healthy subjects and 5 older and healthy subjects\footnote{This study was approved by the UC Berkeley Center for Protection of Human Subjects, Protocol \#2015-07-7767 and informed consent was obtained from all subjects.}.  
Data for individual subjects is shown in Table \ref{tab:subjects}. 
\begin{table}
\centering
  \begin{tabular}{c|c|c|c|c|c}
    \toprule
    Group & ID & Gender & Age & Height (cm) & Weight (kg) \\ \hline
    \multirow{10}{*}{Young}  
    & 1 & F & 23 & 153.7 & 70.3 \\ \cline{2-6}
    & 2 & F & 25 & 165.1 & 68.6 \\ \cline{2-6}
    & 4 & M & 37 & 184.2 & 74.0 \\ \cline{2-6}
    & 5 & M & 26 & 180.3 & 66.2 \\ \cline{2-6}
    & 6 & F & 22 & 165.1 & 58.7 \\ \cline{2-6}
    & 7 & M & 28 & 175.3 & 55.1 \\ \cline{2-6}
    & 8 & M & 29 & 175.3 & 79.8 \\ \cline{2-6}
    & 9 & M & 21 & 167.6 & 64.9 \\ \cline{2-6}
    & 10 & M & 25 & 172.7 & 69.1 \\ \cline{1-6}
    & 14 & F & 25 & 160.0 & 54.2 \\ \cline{2-6}
    \multirow{5}{*}{Older} 
    & 3 & F & 84 & 162.6 & 65.3 \\ \cline{2-6}
    & 11 & M & 69 & 185.4 & 92.5 \\ \cline{2-6}
    & 12 & M & 77 & 170.2 & 66.2 \\ \cline{2-6}
    & 13 & F & 76 & 164.7 & 64.6 \\ \cline{2-6}
    & 15 & M & 74 & 175.3 & 69.1 \\
    \bottomrule
  \end{tabular}
  \caption{Data for each subject}
  \label{tab:subjects}
\end{table}
Initially, subjects were asked to stand without instruction to record the natural STS strategies at slow and fast speeds (`Untrained' dataset). Subsequently, subjects were shown videos demonstrating the momentum transfer and quasi-static STS strategies to avoid individual interpretation of the motion. 
The subjects were then asked to perform the momentum transfer and quasi-static STS strategies in a randomised order (`Trained' dataset')\footnote{Videos can be found at: \url{https://www.w3id.org/people/vshia/jrsi}}.

\subsection{Standing Models}

The inverted pendulum model (IPM), shown in Figure \ref{fig:ipm}, and the double inverted pendulum (DPM), shown in Figure \ref{fig:dpm} are the controlled dynamic models for STS investigated in this paper.
Although the DPM is a more accurate representation of human morphology, IPM has more widespread use due to the complexity of DPM.

The IPM consists of an inverted pendulum attached to a fixed foot on the ground with the point mass $m$ at length $l$ away from the joint.  
\begin{figure}
\begin{subfigure}[b]{0.24\textwidth}
\centering
\includegraphics[width=\textwidth,clip=true,trim=5cm 1.5cm 4cm 1cm]{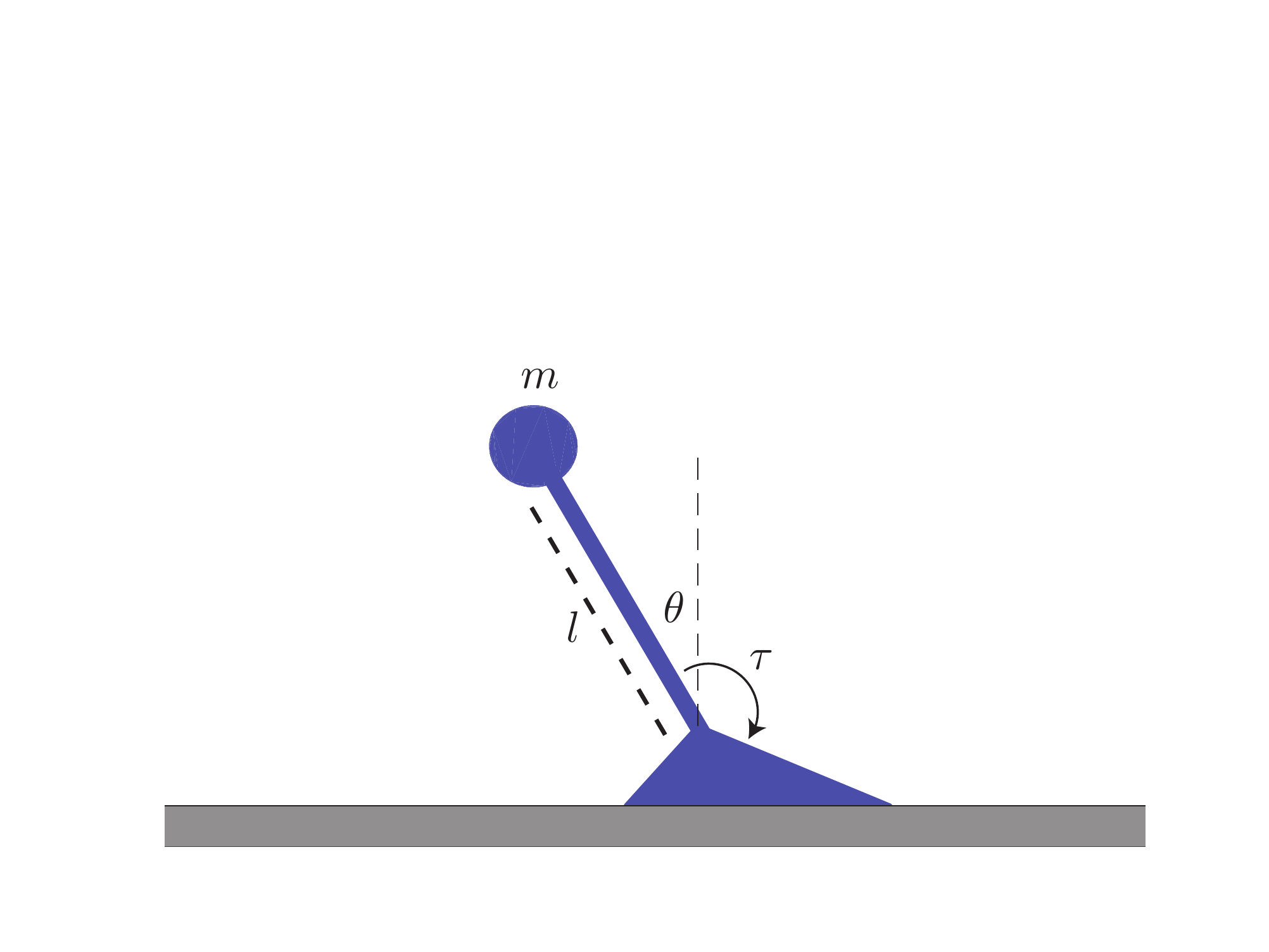}
\caption{Inverted Pendulum Model (IPM)}
\label{fig:ipm}
\end{subfigure}
\begin{subfigure}[b]{0.24\textwidth}
\centering
\includegraphics[width=\textwidth,clip=true,trim=5cm 1.5cm 4cm 1cm]{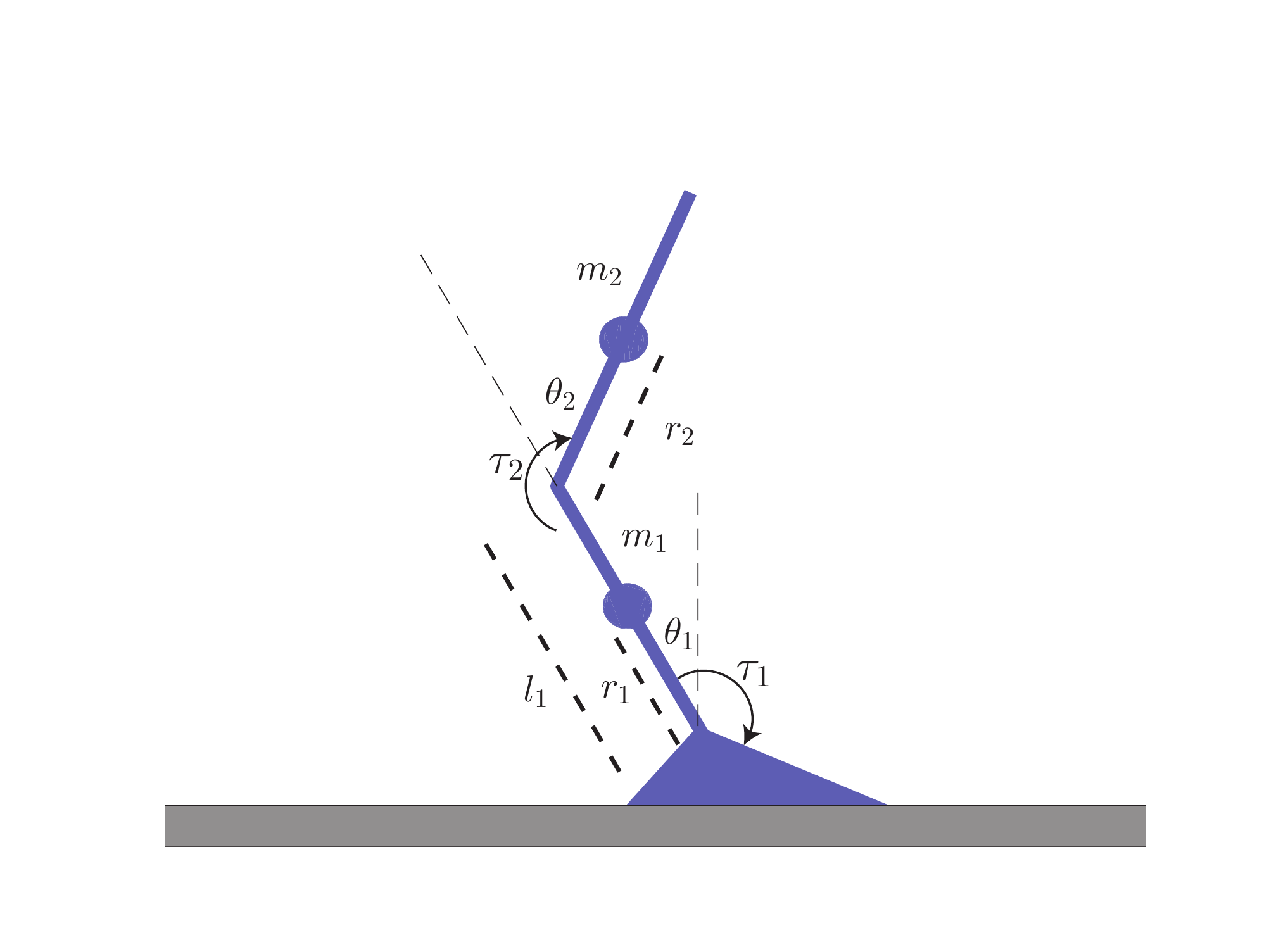}
\caption{Double Pendulum Model (DPM)}
\label{fig:dpm}
\end{subfigure}
\caption{The pair of STS controlled dynamic models considered in this paper.}
\end{figure}
Let $\theta$ represent the angle (with respect to the vertical) and $\dot{\theta}$ represent angular velocity of the pendulum. 
Together both variables are the state space of the IPM, while $\tau$ represents the actuation at the ankle.
The dynamics of the IPM can be found in ~\cite{Pai1997}.
A $5^\text{th}$ order Taylor expansion of the dynamics is used while solving for the BOS using (D).
The motion capture of each individual was fit to the IPM by setting $m$ as the subject's mass, $l$ as the average distance from the subject's ankle to the subject's COM, and $\theta$ as the angle from the ankle to the subject's COM.  

The DPM is a double inverted pendulum attached to a fixed foot on the ground.  
Let $\theta_1$ and $\dot\theta_1$ represent the angle and angular velocity of the lower link to the vertical and $\theta_2$ and $\dot\theta_2$ represent the angle and angular velocity of the lower to upper link shown in Figure \ref{fig:dpm}.  Together these variables represent the state space of the DPM while $\tau_1$ and $ \tau_2$ represent the ankle and hip actuation, respectively.
The dynamics of the DPM can be found in \cite{DPM_model}.

The motion capture of each individual was fit to the DPM by setting  $m_1$ as the mass of the subject's lower body (calf and thigh), $m_2$ as the mass of the subject's upper body, $l_1$ as the average length of the subject's ankle to hip, $r_1$ as the average length from the ankle to the COM of the lower body, and $r_2$ as the average length from the hip to the COM of the upper body.
$\theta_1$ represents the angle from the subject's ankle to hip and $\theta_2$ represents the angle of the subject's hip to upper body.

For both models, masses and the COM positions of each individual limb were computed using tabulated values found in \cite{Pavol2002}.
Individualised torque bounds are set to the maximum and minimum torques obtained via the optimal control.
The domain bounds are the minimum and maximum observed values from the data.
\section{Results}\label{sec:results}
In this section we compare an existing method of estimating stability to the framework proposed above.
All analysis was performed on a system with an Intel Core i7-4930K 3.40GHz processor with 12 cores and 32 GB RAM.
The optimal control problem was solved using MATLAB's nonlinear solver \textit{fmincon}~\cite{matlab_optimization_toolbox}.
The optimization problem ($D$) is solved using SPOTLESS~\cite{spotless} and MOSEK~\cite{mosek}.
Code and figures may be found at: \url{https://www.w3id.org/people/vshia/jrsi}.

\subsection{Results using an existing stability metric}\label{sec:fujimoto}
The existing model-based approach for determining the BOS of STS is called Region of Stability based on Velocity (ROSv). 
This method plots the normalised position and velocity\footnote{Position normalised to the subject's foot length and velocity normalised to pendulum length / second.} of the subject's COM at the instant they rise from the chair (as in Figure \ref{fig:fujimoto})~\cite{Pai2003,Fujimoto2013}. 
Points left of the black line indicate insufficient velocity to stand, whereas points right of the dashed line indicate a catastrophic fall forward. The distance to the dashed line is used to measure stability, with larger distances indicating higher stability.

Table \ref{tab:metric_summary} shows the median ROSv value across 5 trials for each STS motion.  In both cohorts, ROSv overall determines that slow and quasi-static manoeuvres are more stable than the fast and momentum-transfer manoeuvres, respectively, which is congruent with the intuition developed in Section \ref{subsec:motor}.
Upon further examination, ROSv is most unreliable when computing the stability of young subjects performing slow and fast STS motions. Figure \ref{fig:fujimoto} shows the ROSv plot for a young subject (a. ID 7) and an older subject (b. ID 11) that are inaccurately characterised by ROSv.
Hereafter, we continue to highlight subjects ID 7 and 11 to compare the accuracy of each method.
These results illustrate the deficiencies of estimating the stability of motion with a single feature and the inability of the ROSv metric to characterise the specific perturbations that lead to a fall.

\begin{figure}[h]
  \centering
  \begin{subfigure}[b]{0.24\textwidth}
    \centering
    \includegraphics[width=\textwidth,trim={0cm 0cm 0cm 0cm},clip]{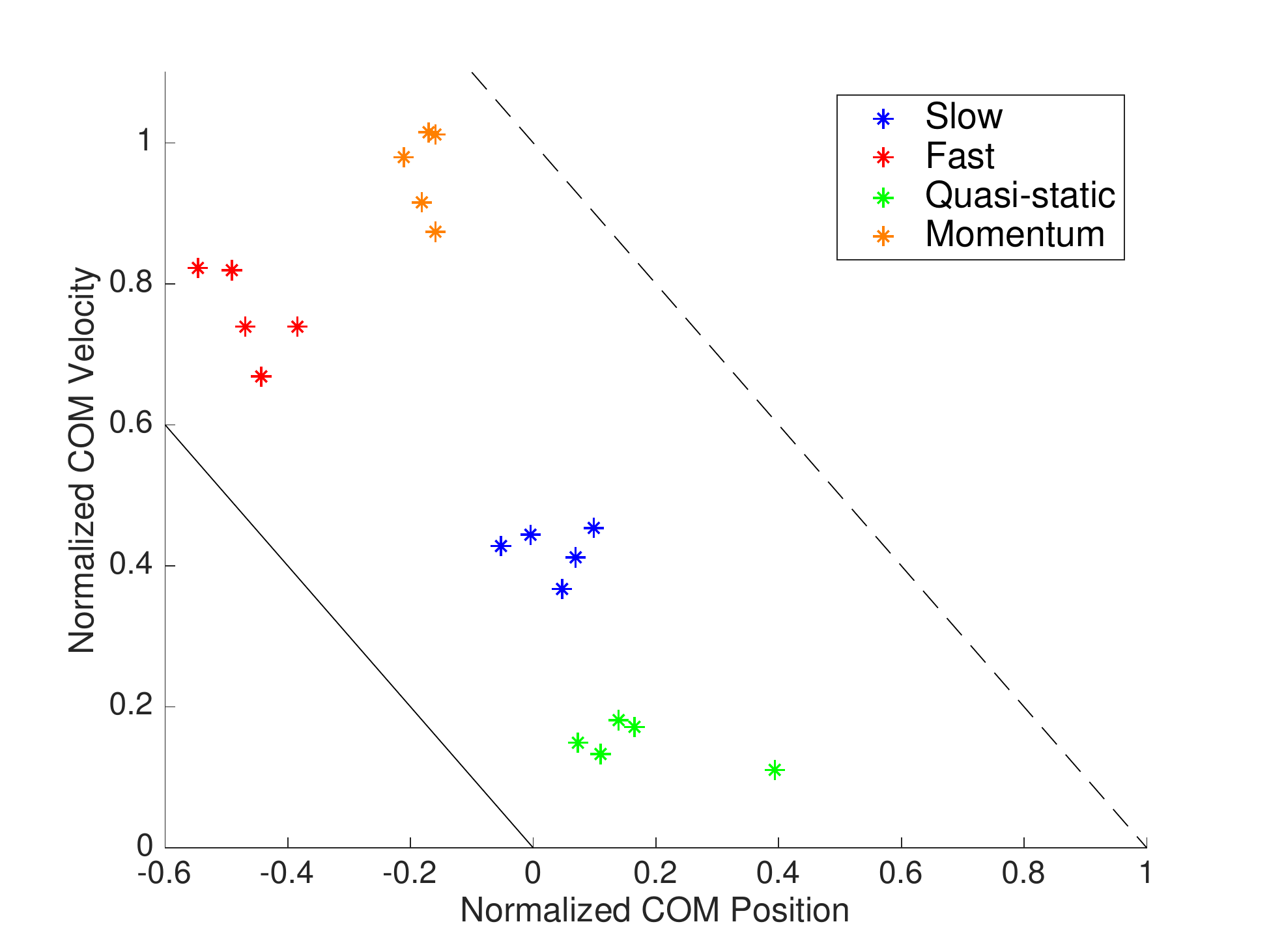}
    \caption{Young Subject (ID 7)}
    \label{fig:fujimoto_a}
  \end{subfigure}
  \begin{subfigure}[b]{0.24\textwidth}
    \centering
    \includegraphics[width=\textwidth,trim={0cm 0cm 0cm 0cm},clip]{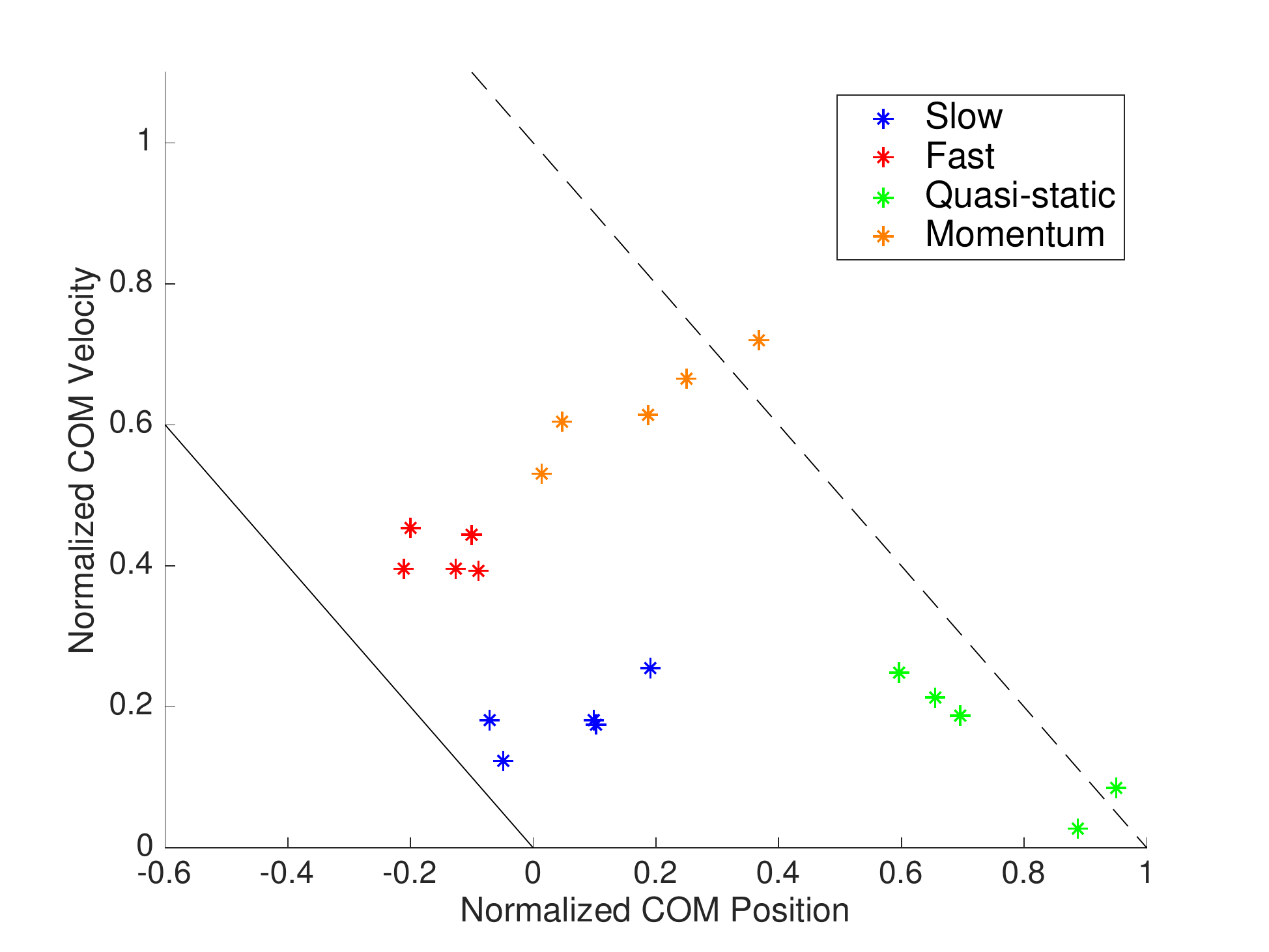}
    \caption{Older subject (ID 11)}
    \label{fig:fujimoto_b}
  \end{subfigure}
  \caption{ROSv plots of an individual's STS motions. The stars denote individual trials for each motion.  The area between the black and dashed lines represent initial positions and velocities where subjects have enough torque to stand. The area to the right of the dashed line indicates falling forward, and the area to the left of the black line indicates remaining seated. 
}
  \label{fig:fujimoto}
\end{figure}

\subsection{Results for IPM}
Using the method proposed in Section \ref{sec:method}, the BOS for each individual's motion is computed for the IPM.
The optimal control was performed using $101$ time steps with a polynomial input of degree $6$, and the optimization problem ($D$) was solved with a degree $14$ polynomial.
The entire pipeline for a single action required an average of $211$ seconds to compute.

To compare the computed BOS for trajectories of different time lengths, the BOS volume is normalised by the volume of the domain with bounds.
For example, 100\% indicates that the bounded domain is in the BOS and 0\% indicates that the BOS is empty.
Table \ref{tab:ipm_metrics} shows the median of the normalised volumes for the computed BOS of each subjects over all trials of a specific manoeuvre. 

All slower and quasi-static STS motions have larger basins than faster and momentum-transfer STS motions, respectively, indicating that subjects who use slower and more static motions are able to withstand more perturbations.
The proposed method correctly determines the STS strategies with higher stability according to the intuition developed in Section \ref{sec:experiments}.

\begin{table}[h]
  \centering
  \begin{tabular}{c|c|cccc}
  \toprule
    \multirow{2}{*}{Group} & \multirow{2}{*}{ID} & \multicolumn{2}{c}{Untrained Motion} & \multicolumn{2}{c}{Trained STS Strategy} \\ \cline{3-6}
    & & Slow & Fast & Quasi-static & Momentum \\ \hline
    \multirow{10}{*}{Young} 
& 1 & 35.9 & 23.0 & 34.5 & 11.9  \\ \cline{2-6} 
& 2 & 42.9 & 17.2 & 46.9 & 15.9  \\ \cline{2-6}  
& 4 & 38.5 & 16.8 & 30.4 & 13.9  \\ \cline{2-6}  
& 5 & 38.1 & 18.4 & 38.1 & 19.3  \\ \cline{2-6}  
& 6 & 41.0 & 9.7 & 45.1 & 8.7  \\ \cline{2-6}  
& \textbf{7} & \textbf{31.2} & \textbf{16.5} & 37.5 & 15.0  \\ \cline{2-6}  
& 8 & 37.3 & 21.9 & 37.5 & 17.2  \\ \cline{2-6}  
& 9 & 35.4 & 18.5 & 38.7 & 14.3  \\ \cline{2-6}  
& 10 & 36.1 & 10.7 & 38.9 & 8.8  \\ \cline{2-6}  
& 14 & 37.7 & 16.7 & 35.3 & 11.6  \\ \cline{1-6}  
    \multirow{5}{*}{Older} 
& 3 & 38.8 & 26.1 & 42.2 & 23.4  \\ \cline{2-6}  
& \textbf{11} & 40.3 & 21.0 & \textbf{45.5} & \textbf{14.5}  \\ \cline{2-6}  
& 12 & 35.5 & 19.1 & 34.6 & 21.1  \\ \cline{2-6}  
& 13 & 51.2 & 23.4 & 53.0 & 9.9  \\ \cline{2-6}  
& 15 & 45.9 & 16.8 & 40.6 & 12.0  \\ \cline{2-6}  
    \bottomrule
  \end{tabular}
  \caption{Median volume of the BOS for each STS strategy calculated using the IPM, as a percentage of the domain with bounds. The individuals and strategies illustrated in Figure \ref{fig:fujimoto} are shown in bold.}
  \label{tab:ipm_metrics}
\end{table}

The shape of the BOS in Figures \ref{fig:bos_ipm_young} and \ref{fig:bos_ipm_older} indicates the perturbations the individuals in Figure \ref{fig:fujimoto} are able to withstand under a specific control.
Notice that the BOS for the quasi-static strategy is larger, indicating greater stability, at the onset of the STS action.
These results demonstrate that the proposed method succeeds in cases that ROSv fails to properly assess, and provides further information regarding sources of instability throughout the STS action.

\begin{figure*}[h]
  \begin{subfigure}[b]{0.48\textwidth}
    \centering
    \caption{Young subject - Slow}    
    \includegraphics[width=\textwidth,trim={1cm 0.5cm 0.75cm 1cm},clip]{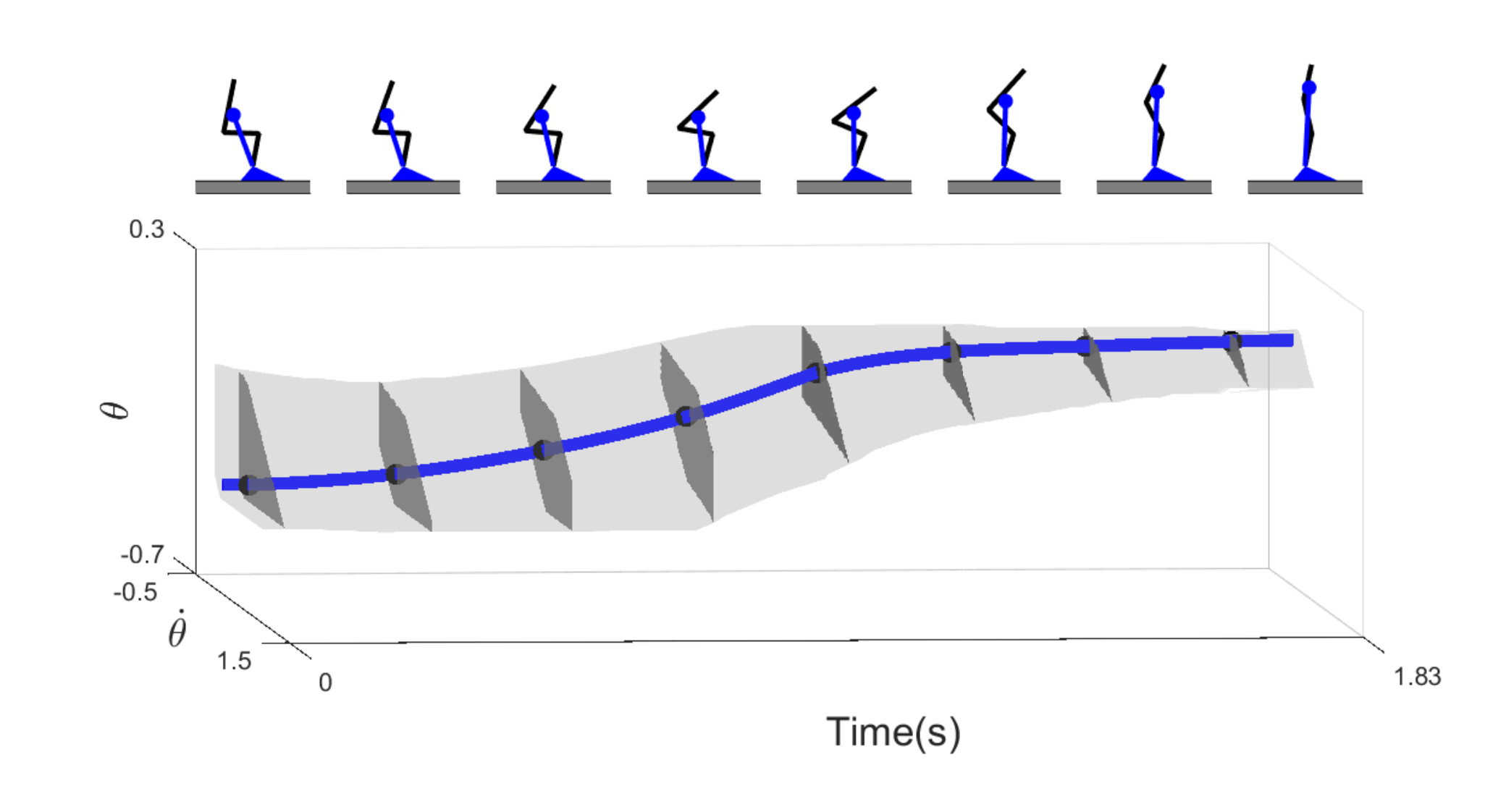}
    \label{fig:bos_ipm_a}
  \end{subfigure}
  \begin{subfigure}[b]{0.48\textwidth}
    \centering
    \caption{Young subject - Fast}
    \includegraphics[width=\textwidth,trim={1cm 0.5cm 0.75cm 1cm},clip]{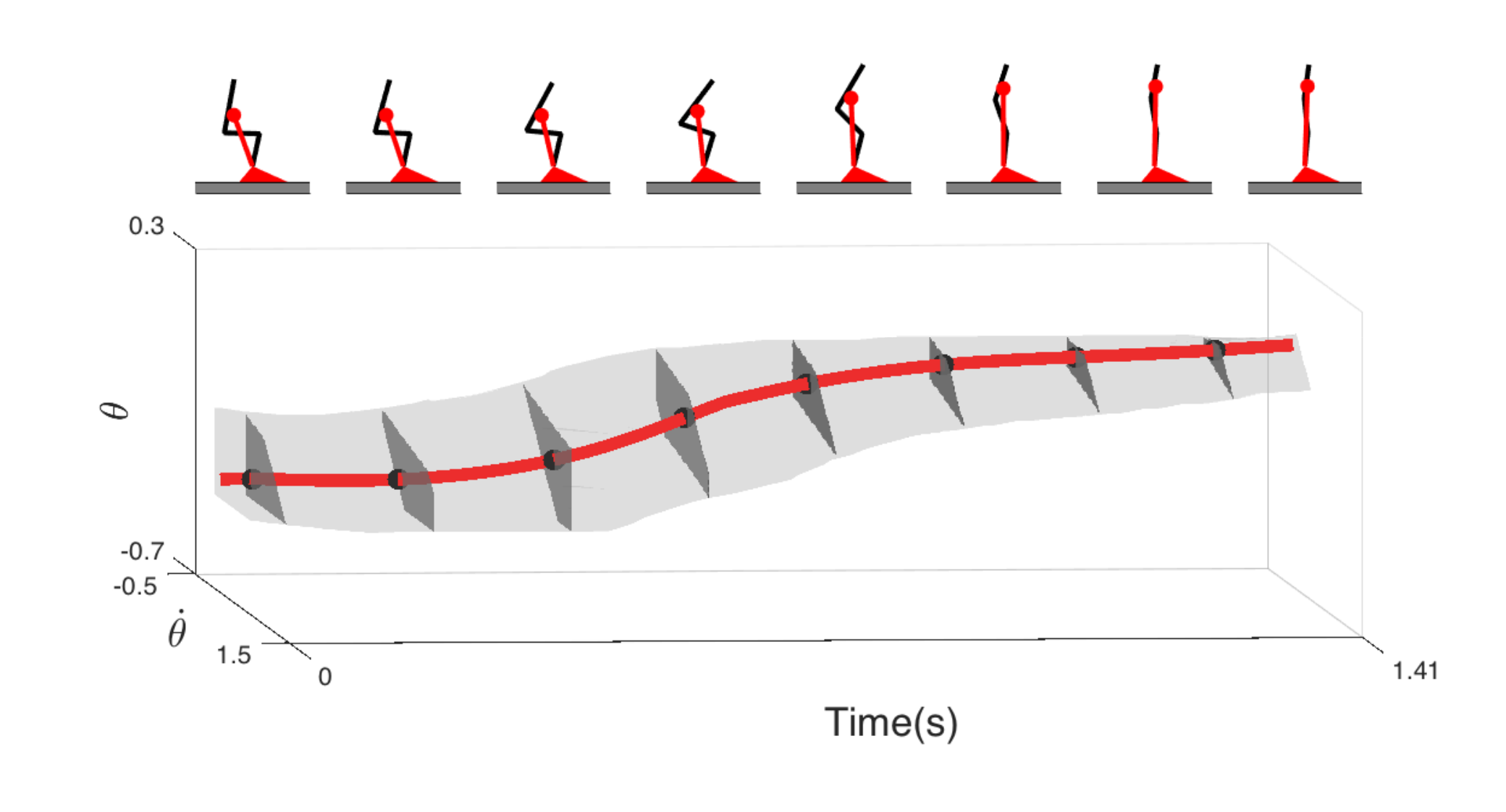}
    \label{fig:bos_ipm_b}
  \end{subfigure}
    \caption{The computed BOS for the IPM using the proposed method for untrained STS motions of subject ID 7. The top panels show a diagram of the STS motion through time, with the non-gray colour depicting the IPM representation of the motion. The bottom panels show the computed BOS for the IPM in the state space of the model. Thick coloured lines represent the observed trajectory ($x_\text{obs}$) of the STS motion. The gray region represents the computed BOS with time slices in dark gray corresponding to the diagram above. Time is on the X-axis (scaled in each subplot), angle from foot to COM on the Y-axis, and angular velocity on the Z-axis.}
  \label{fig:bos_ipm_young}
\end{figure*}

\begin{figure*}[h]
  \begin{subfigure}[b]{0.48\textwidth}
    \centering
    \caption{Older subject - Quasi-static}
    \includegraphics[width=\textwidth,trim={1cm 0.5cm 0.75cm 1cm},clip]{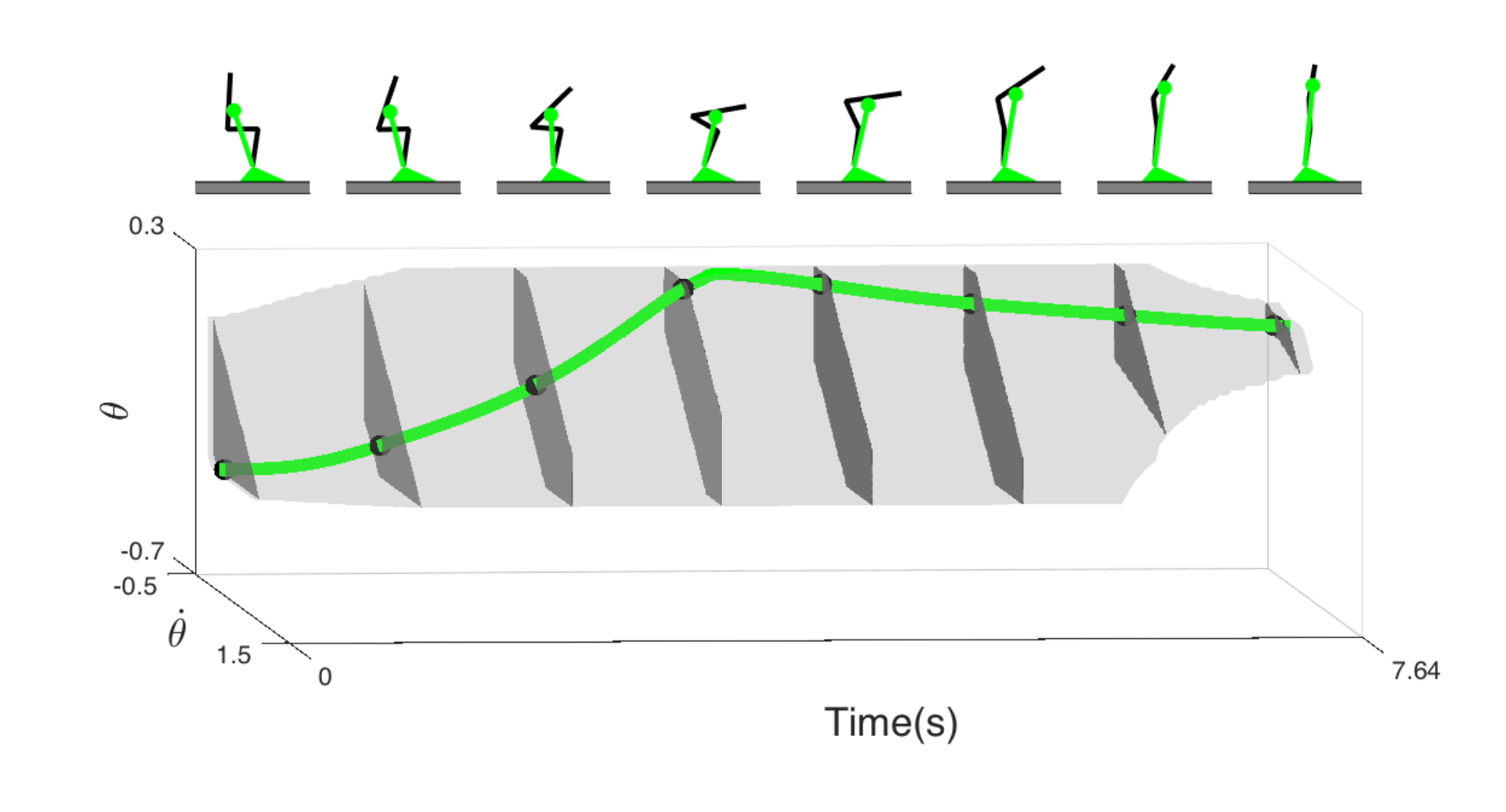}
    \label{fig:bos_ipm_g}
  \end{subfigure}
  \begin{subfigure}[b]{0.48\textwidth}
    \centering
    \caption{Older subject - Momentum Transfer}
    \includegraphics[width=\textwidth,trim={1cm 0.5cm 0.75cm 1cm},clip]{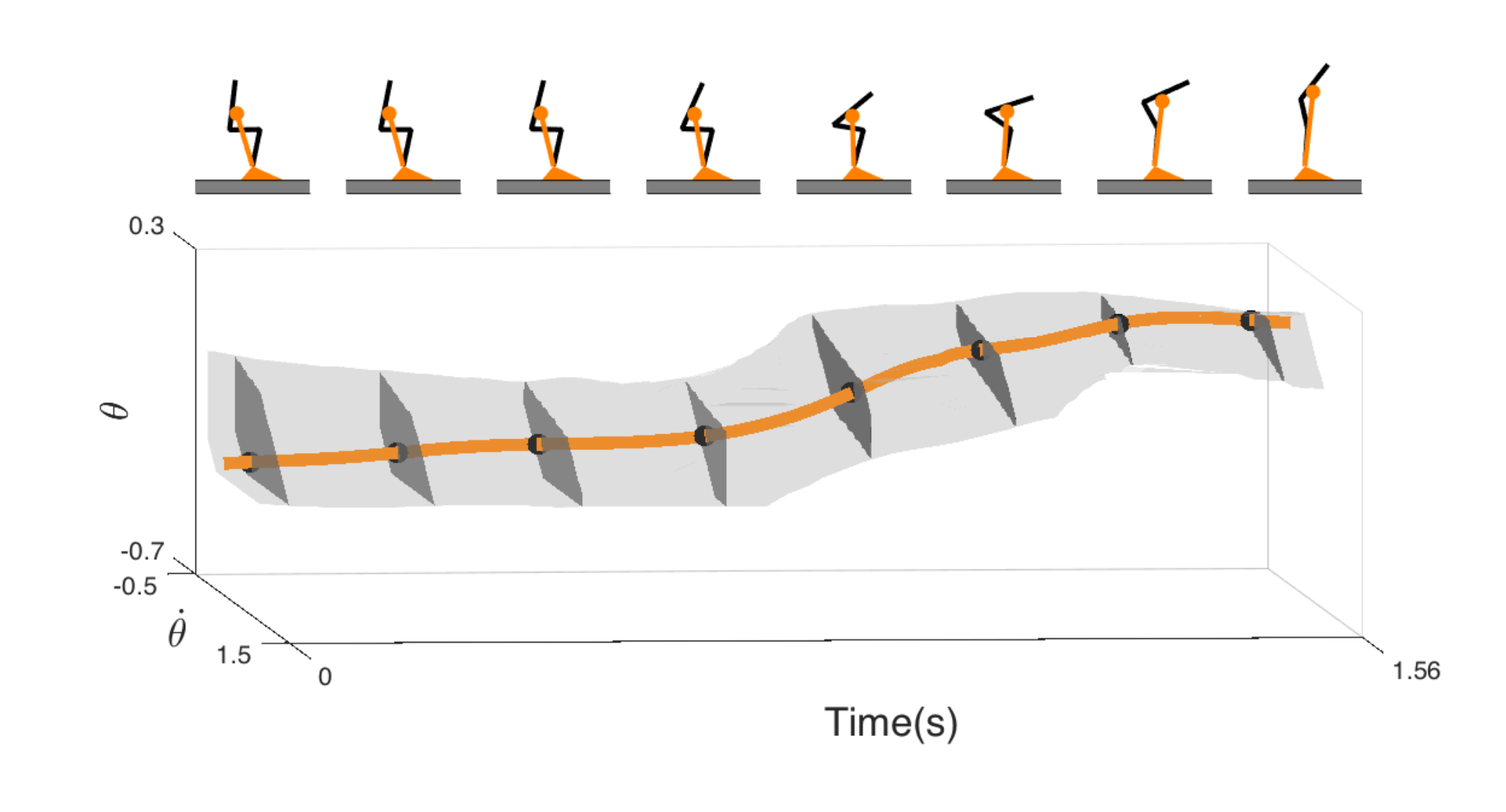}
    \label{fig:bos_ipm_h}
  \end{subfigure}
  \caption{The computed BOS for the IPM using the proposed method for trained STS motions of subject ID 11. The top panels show a diagram of the STS motion through time, with the non-gray colour depicting the IPM representation of the motion. The bottom panels show the computed BOS for the IPM in the state space of the model. Thick coloured lines represent the observed trajectory ($x_\text{obs}$) of the STS motion. The gray region represents the computed BOS with time slices in dark gray corresponding to the diagram above. Time is on the X-axis (scaled in each subplot), angle from foot to COM on the Y-axis, and angular velocity on the Z-axis.}
  \label{fig:bos_ipm_older}
\end{figure*}

\subsection{Results for DPM}
To determine the effect of using a model that more accurately reflects the morphology of an individual, the BOS for each individual's motion is computed for the DPM.
The optimal control was performed using $101$ time steps with a polynomial input of degree $4$, and the optimization problem ($D$) was solved with a degree $8$ polynomial.
The entire pipeline for a single action required an average of $4130$ seconds to compute.
Table \ref{tab:dpm_metrics} describes the volume of the BOS normalised by the volume of the bounded domain using the the DPM for all subjects.  

\begin{figure*}[h]
  \begin{subfigure}[b]{0.48\textwidth}
    \centering
    \caption{Young subject - Slow}    
    \includegraphics[width=\textwidth,trim={1cm 0.5cm 0.75cm 1cm},clip]{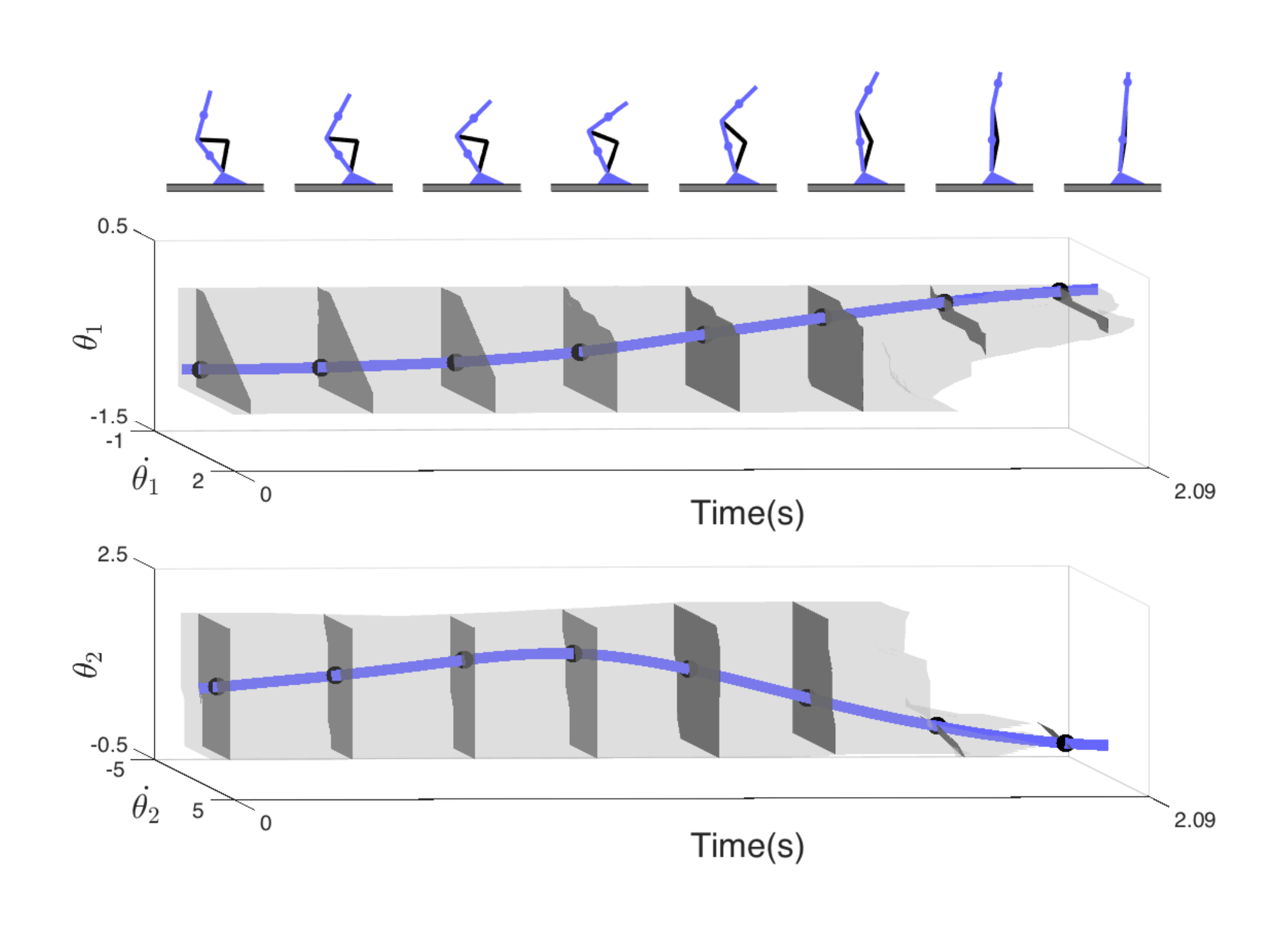}
    \label{fig:bos_dpm_a}
  \end{subfigure}
  \begin{subfigure}[b]{0.48\textwidth}
    \centering
    \caption{Young subject - Fast}
    \includegraphics[width=\textwidth,trim={1cm 0.5cm 0.75cm 1cm},clip]{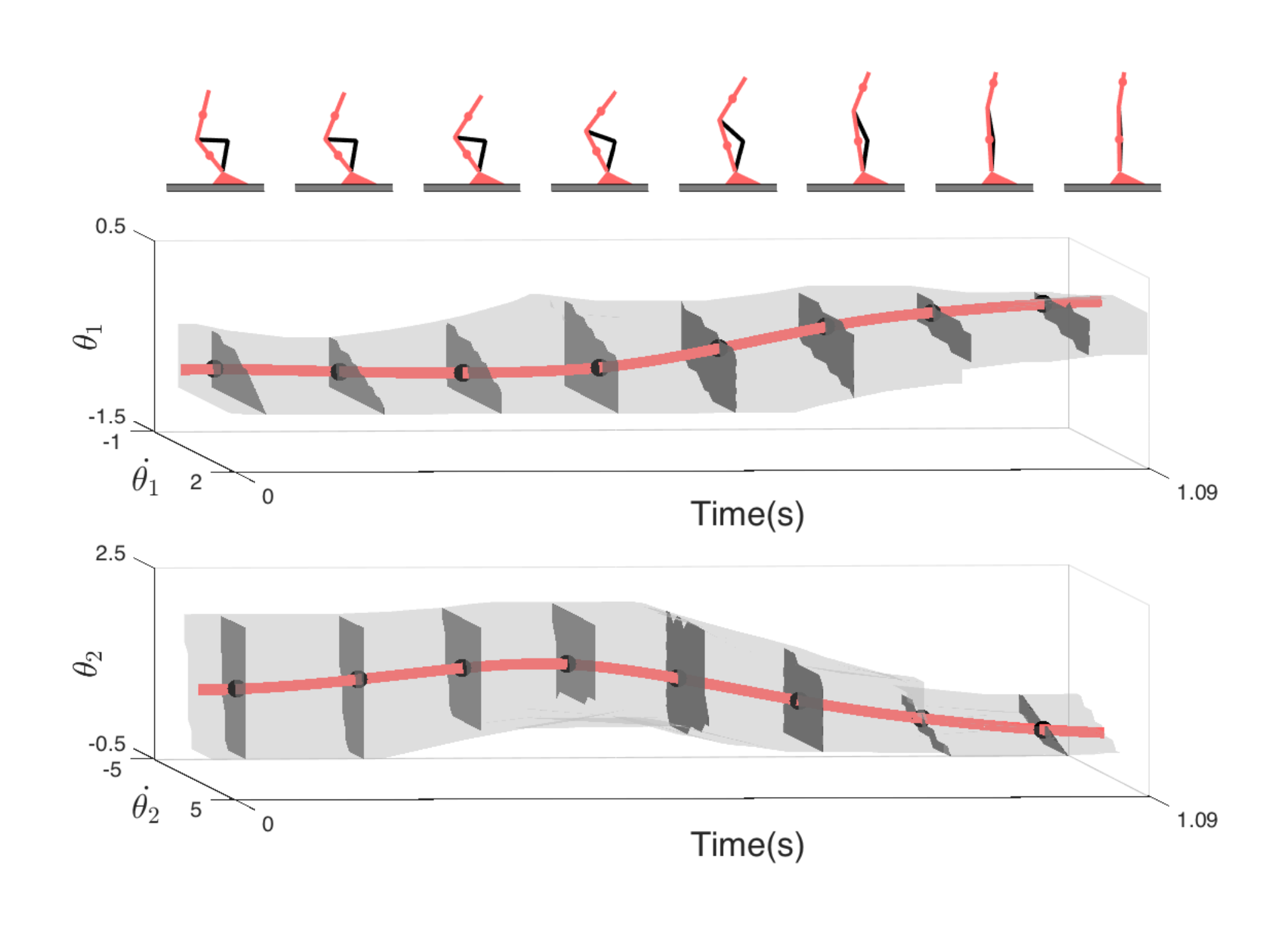}
    \label{fig:bos_dpm_b}
  \end{subfigure}
    \caption{The computed BOS for the DPM for untrained STS motions of subject ID 7. The top panels illustrate the STS motion through time, with the coloured skeleton depicting the DPM representation of the motion. The center and bottom panels represent the BOS for each joint of the DPM in the state space of the model. Thick coloured lines in each subplot represent the observed trajectory ($x_\text{obs}$) of the STS motion. The gray region represents the computed BOS, with time slices in dark gray corresponding to the position of the model above.  Time is on the X-axis (scaled in each subplot), angle from the joint (1 or 2) to the limb COM ($m_1$ or $m_2$) on the Y-axis, and corresponding joint angular velocity on the Z-axis.}
  \label{fig:bos_dpm_young}
\end{figure*}

\begin{figure*}[h]
  \begin{subfigure}[b]{0.48\textwidth}
    \centering
    \caption{Older subject - Quasi-static}
    \includegraphics[width=\textwidth,trim={1cm 0.5cm 0.75cm 1cm},clip]{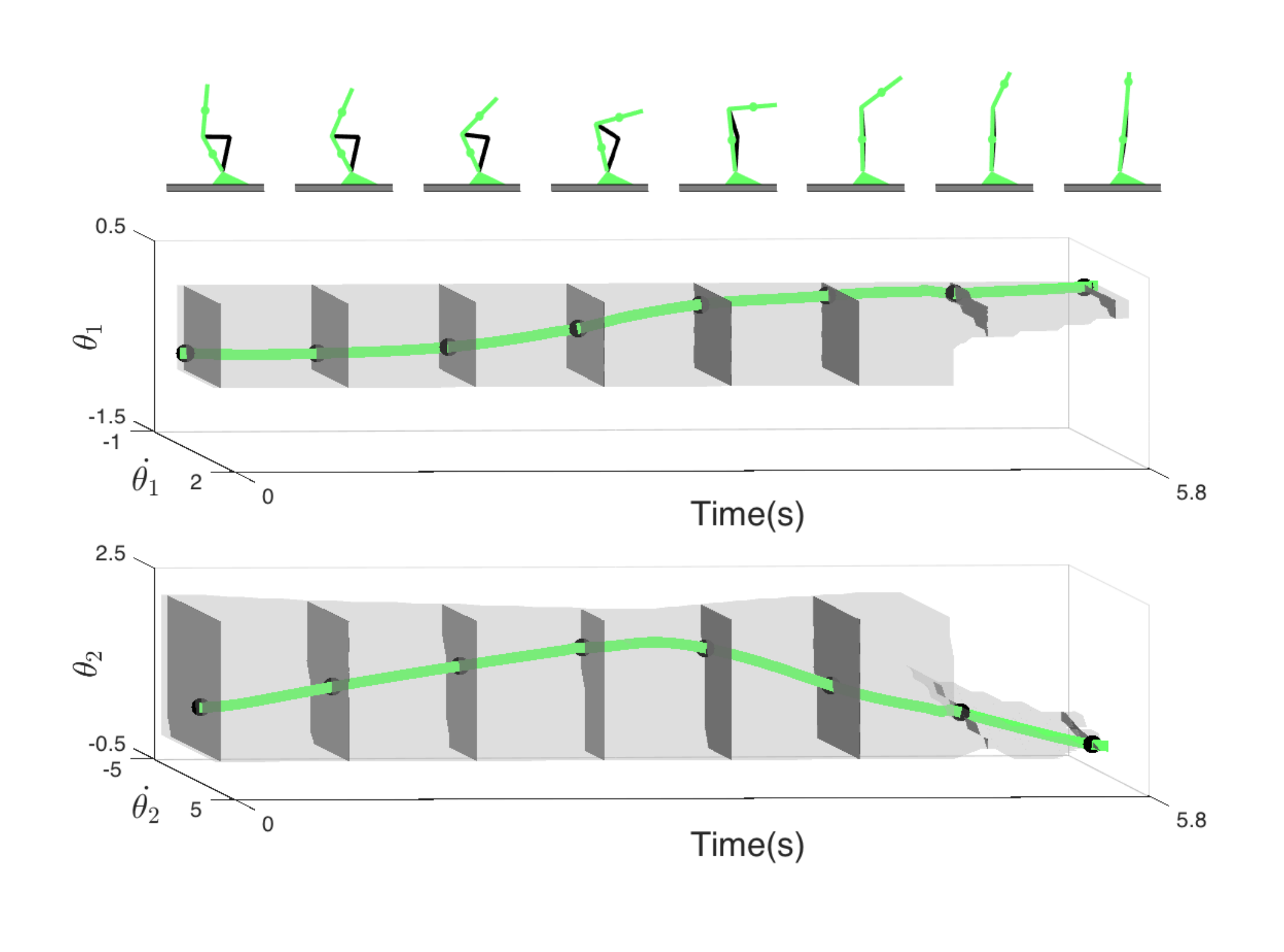}
    \label{fig:bos_dpm_g}
  \end{subfigure}
  \begin{subfigure}[b]{0.48\textwidth}
    \centering
    \caption{Older subject - Momentum Transfer}
    \includegraphics[width=\textwidth,trim={1cm 0.5cm 0.75cm 1cm},clip]{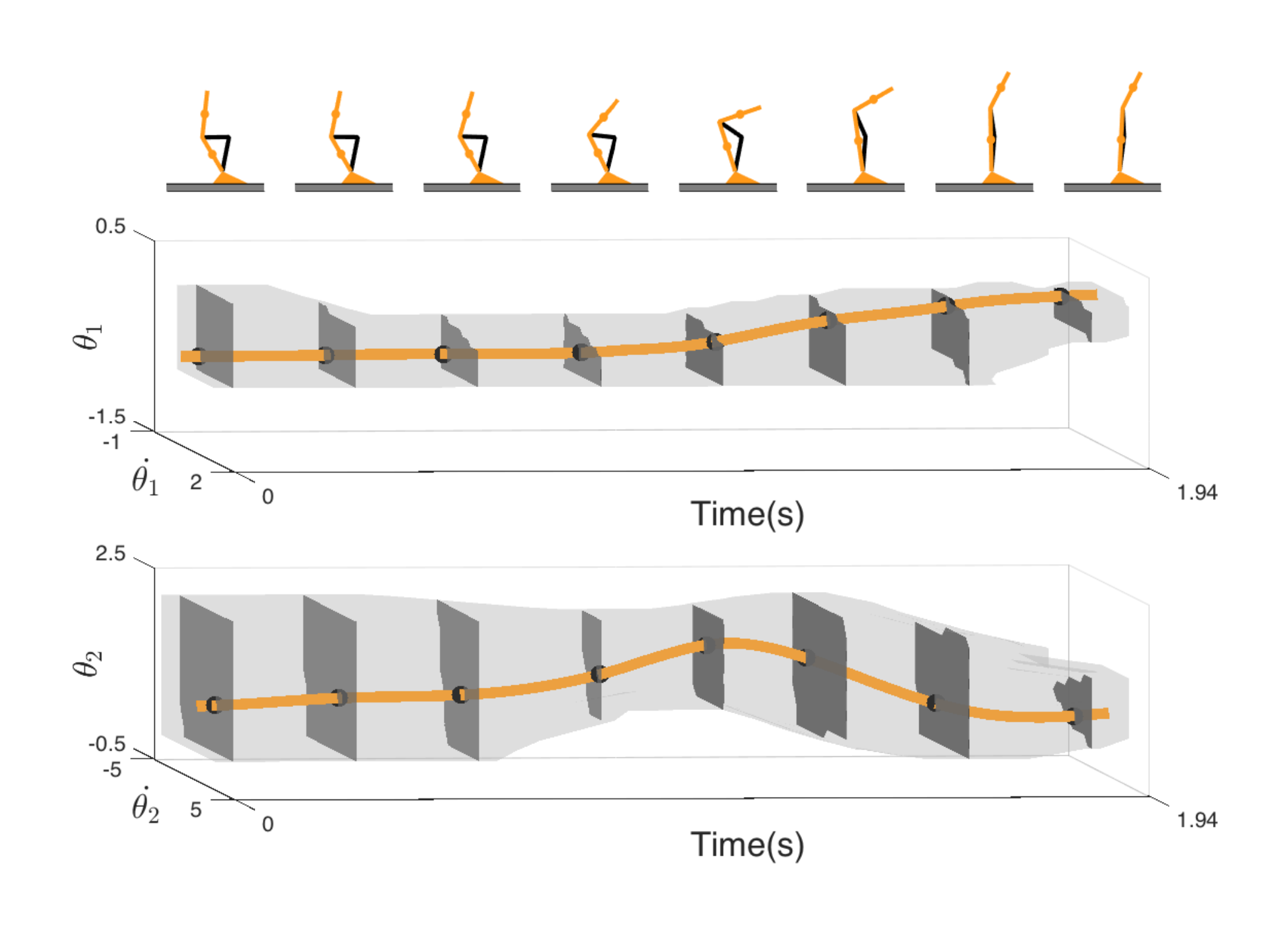}
    \label{fig:bos_dpm_h}
  \end{subfigure}
  \caption{ The computed BOS for the DPM for trained STS motions of subject ID 11. The top panels illustrate the STS motion through time, with the coloured skeleton depicting the DPM representation of the motion. The center and bottom panels represent the BOS for each joint of the DPM in the state space of the model. Thick coloured lines in each subplot represent the observed trajectory ($x_\text{obs}$) of the STS motion. The gray region represents the computed BOS, with time slices in dark gray corresponding to the position of the model above.  Time is on the X-axis (scaled in each subplot), angle from the joint (1 or 2) to the limb COM ($m_1$ or $m_2$) on the Y-axis, and corresponding joint angular velocity on the Z-axis.}
  \label{fig:bos_dpm_older}
\end{figure*}

\begin{table}[h]
  \centering
  \begin{tabular}{c|c|cccc}
  \toprule
    \multirow{2}{*}{Group} & \multirow{2}{*}{ID} & \multicolumn{2}{c}{Untrained Motion} & \multicolumn{2}{c}{Trained STS Strategy} \\ \cline{3-6}
    & & Slow & Fast & Quasi-static & Momentum \\ \hline
    \multirow{10}{*}{Young} 
& 1 & 16.9 & 11.4 & 20.8 & 11.1  \\ \cline{2-6} 
& 2 & 21.1 & 11.0 & 30.6 & 17.0  \\ \cline{2-6}  
& 4 & 13.1 & 10.2 & 14.6 & 10.8  \\ \cline{2-6}  
& 5 & 12.9 & 8.2 & 15.2 & 8.4  \\ \cline{2-6}  
& 6 & 17.8 & 16.3 & 22.8 & 13.2  \\ \cline{2-6}  
& \textbf{7} & \textbf{18.4} & \textbf{12.2} & 19.5 & 12.6  \\ \cline{2-6}  
& 8 & 13.5 & 7.5 & 13.8 & 7.0  \\ \cline{2-6}  
& 9 & 27.6 & 17.4 & 26.3 & 19.3  \\ \cline{2-6} 
& 10 & 14.8 & 9.3 & 17.0 & 11.3  \\ \cline{2-6} 
& 14 & 23.0 & 14.7 & 24.6 & 11.2  \\ \cline{1-6} 
\multirow{5}{*}{Older} 
& 3 & 36.5 & 30.6 & 32.7 & 29.8  \\ \cline{2-6}  
& \textbf{11} & 28.1 & 20.2 & \textbf{27.6} & \textbf{22.1}  \\ \cline{2-6} 
& 12 & 8.5 & 5.9 & 9.8 & 6.7  \\ \cline{2-6}  
& 13 & 18.0 & 15.4 & 26.3 & 14.8  \\ \cline{2-6}  
& 15 & 24.1 & 12.0 & 19.6 & 11.3  \\ \cline{2-6}  
    \bottomrule
  \end{tabular}
  \caption{Median volume of the BOS for each STS strategy calculated using the DPM, as a percentage of the domain with bounds. The individuals and strategies illustrated in Fig. \ref{fig:fujimoto} are shown in bold.}
  \label{tab:dpm_metrics}
\end{table}

The BOS for subjects ID 7 and 11 are illustrated in Figures \ref{fig:bos_dpm_young} and \ref{fig:bos_dpm_older}, respectively.
Again, note for subject ID 7, the BOS for the slow STS is larger than the fast STS, and for subject ID 11, the BOS for quasi-static STS is larger than the momentum-transfer STS. Much like the IPM, the DPM also succeeds in cases when ROSv fails.
The BOS holds a rectangular shape at certain times due to the $\theta_1$ and $\theta_2$ bounds and indicating that the subject is ``maximally'' stable with respect to those states.

\subsection{Summarizing Performance}

As shown in Table \ref{tab:metric_summary}, the BOS computation method presented above correctly identifies the STS strategies with greater stability with a higher accuracy than the ROSv method. 
Furthermore, by increasing model complexity from the IPM to DPM, the volume of the computed BOS tends to decrease, suggesting that higher order models may yield tighter BOSs about the trajectory.

\begin{table}[h]
  \centering
  \begin{tabular}{c|c|cccc}
  \toprule
    \multirow{2}{*}{Method} & \multirow{2}{*}{Type} & \multicolumn{2}{c}{Slow $>$ Fast} & \multicolumn{2}{c}{Static $>$ Momentum} \\ \cline{3-6}
    & & Young & Older & Young & Older \\ \hline
    Fujimoto & ROSv & 6/10 & 4/5 & 9/10 & 4/5 \\ \hline 
    \multirow{2}{*}{Proposed} & IPM & 10/10 & 5/5 & 10/10 & 5/5 \\ \cline{2-6} 
    & DPM & 10/10 & 5/5 & 10/10 & 5/5 \\
    \bottomrule
  \end{tabular}
  \caption{The number of individuals for which each method correctly identified the more stable standing strategy (Slow and Static).  The median value for ROSv and the BOS volume across 5 trials for each STS motion was used for the comparison.}
  \label{tab:metric_summary}
\end{table}
\section{Discussion and Conclusion}\label{sec:discussion}

This paper presents the first personalised computational framework to model, identify, and analyse the stability of an individual's STS motion by using kinematic observations to compute the BOS.
Rather than reducing the STS motion to a single feature, the entire trajectory is analysed via subject-specific models and motion-specific trajectories to provide a more informative metric of stability.
Where ROSv method fails, our proposed method successfully identifies slow and quasi-static as more stable than fast and momentum-transfer standing strategies. The shape of the computed BOS reveals how stability changes throughout the STS motion, aiding in the identification of unstable manoeuvres. 

This framework provides a clinical tool to aid physical therapists in identifying and reducing locomotor instability. Using numerical tools to compute the BOS of locomotion obviates the need to perform extensive perturbation experiments, making this approach applicable to injured or high-risk individuals. Double blind tests comparing the motion of healthy, injured, and fall-prone individuals will help establish the diagnostic benefit of this automated computation of stability. Because the shape of the BOS characterises the perturbations most detrimental to an individual, this method can direct the customization of physical therapy regimens to improve stability. Furthermore, implementing our computational framework into longitudinal studies will improve our understanding of the effect aging, injury, and clinical intervention have on locomotor stability and quality of life.
Fortunately, the speed of this computation method makes it feasible to quickly collect enough trials to experimentally validate these clinical applications with adequate statistical power. 

There is an unavoidable tradeoff between computation speed, accuracy and dimensionality. For systems with few states (i.e. IPM), it is possible to simulate the volume of the BOS via direct simulation and circumvent the optimization problem defined in Section \ref{sec:om}. However,
direct simulation suffers from exponential scaling in the number of states. 
For example, for the DPM, a 2-link pendulum, a sparse simulation of the BOS consists of 175k randomly sampled points requires over 8 hours to compute.
To simulate the BOS for a more representative human model such as a 3-link pendulum or higher would take days or weeks, which is not practically feasible for widespread deployment.
However, equally as important to the BOS is the identification of control strategy used by humans for different actions.

This framework can be expanded to evaluate the stability of a variety of human behaviours, thereby enabling the study of human motion from a control-theoretic point of view. 
Observations of athletes attempting to regain balance suggest that swinging appendages can contribute to overall stability. 
Although it has been demonstrated that the angular momentum of swinging appendages can affect body rotation~\cite{Libby2012}, the evaluation of control strategies exploiting swinging appendages has not occurred.

The proposed method for computing the BOS enables automation, widespread deployment, and customization of motion analysis. 
This framework aids experimental design and analysis of a variety of motions, enhancing the study of individual differences in musculoskeletal architecture and motor control strategies.
The methods presented here can be used to improve the identification of individuals at risk for falling and to develop targeted therapy to increase stability, thereby helping individuals maintain mobility and quality of life.

\section{Acknowledgements}
This work was supported by ONR MURI N00014-13-1-0341.

\section{Competing Interests}
The authors have no competing interests.

\section{Author Contributions}
VS formulated the framework and implementation, conducted human experiments and analysis, and drafted the manuscript; TM assisted in data analysis and helped revise the manuscript for important intellectual content; RB helped conceive the study and draft the manuscript; RV aided in the development of the framework, constructed the tools for stability analysis, and helped draft the manuscript. All authors gave final approval for publication.

\bibliographystyle{ShiaSTSBOS}
{\footnotesize
\bibliography{references}}

\begin{thebibliography}{10}
\expandafter\ifx\csname urlstyle\endcsname\relax
  \providecommand{\doi}[1]{doi:\discretionary{}{}{}#1}\else
  \providecommand{\doi}{doi:\discretionary{}{}{}\begingroup
  \urlstyle{rm}\Url}\fi

\bibitem{CDC2015}
{US CDC}.
\newblock {Important Facts about Falls}.

\bibitem{Robertson2001}
Robertson MC, 2001 {Effectiveness and economic evaluation of a nurse delivered
  home exercise programme to prevent falls. 1: Randomised controlled trial}.
\newblock \emph{BMJ} \textbf{322}:697--697.
\newblock \doi{10.1136/bmj.322.7288.697}.

\bibitem{Horak2006}
Horak FB, 2006 {Postural orientation and equilibrium: what do we need to know
  about neural control of balance to prevent falls?}
\newblock \emph{Age and Ageing} \textbf{35 Suppl 2}:ii7--ii11.
\newblock \doi{10.1093/ageing/afl077}.

\bibitem{Pollock2000}
Pollock A, Durward B, Rowe P, Paul J, 2000 {What is balance?}
\newblock \emph{Clinical Rehabilitation} \textbf{14}:402--406.
\newblock \doi{10.1191/0269215500cr342oa}.

\bibitem{Campbell1989}
Campbell AJ, Borrie MJ, Spears GF, 1989 {Risk factors for falls in a
  community-based prospective study of people 70 years and older.}
\newblock \emph{Journal of Gerontology} \textbf{44}:M112--7.
\newblock \doi{10.1093/geronj/44.4.M112}.

\bibitem{Berg1989}
Berg K, 1989 {Measuring balance in the elderly: preliminary development of an
  instrument}.
\newblock \emph{Physiotherapy Canada} \textbf{41}:304--311.
\newblock \doi{10.3138/ptc.41.6.304}.

\bibitem{Lundin-Olsson1997}
Lundin-Olsson L, Nyberg L, Gustafson Y, 1997 {"Stops walking when talking" as a
  predictor of falls in elderly people.}
\newblock \emph{Lancet} \textbf{349}:617.
\newblock \doi{10.1016/S0140-6736(97)24009-2}.

\bibitem{Podsiadlo1991}
Podsiadlo D, Richardson S, 1991 {The Timed "Up {\&} Go": A Test of Basic
  Functional Mobility for Frail Elderly Persons}.
\newblock \emph{Journal of the American Geriatrics Society}
  \textbf{39}:142--148.
\newblock \doi{10.1111/j.1532-5415.1991.tb01616.x}.

\bibitem{Lundin-Olsson1998}
Lundin-Olsson L, Nyberg L, Gustafson Y, 1998 {Attention, Frailty, and Falls:
  The Effect of a Manual Task on Basic Mobility}.
\newblock \emph{Journal of the American Geriatrics Society}
  \textbf{46}:758--761.
\newblock \doi{10.1111/j.1532-5415.1998.tb03813.x}.

\bibitem{Pai1997}
Pai YC, Patton J, 1997 {Center of mass velocity-position predictions for
  balance control}.
\newblock \emph{Journal of Biomechanics} \textbf{30}:347--354.
\newblock \doi{10.1016/S0021-9290(96)00165-0}.

\bibitem{Papa1999}
Papa E, Cappozzo A, 1999 {A telescopic inverted-pendulum model of the
  musculo-skeletal system and its use for the analysis of the sit-to-stand
  motor task}.
\newblock \emph{Journal of Biomechanics} \textbf{32}:1205--1212.
\newblock \doi{10.1016/S0021-9290(99)00103-7}.

\bibitem{Patton1999}
Patton JL, Pai YC, Lee WA, 1999 {Evaluation of a model that determines the
  stability limits of dynamic balance}.
\newblock \emph{Gait {\&} Posture} \textbf{9}:38--49.
\newblock \doi{10.1016/S0966-6362(98)00037-X}.

\bibitem{Papa2000}
Papa E, Cappozzo A, 2000 {Sit-to-stand motor strategies investigated in
  able-bodied young and elderly subjects}.
\newblock \emph{Journal of Biomechanics} \textbf{33}:1113--1122.
\newblock \doi{10.1016/S0021-9290(00)00046-4}.

\bibitem{Pai2003}
Pai YC, Wening JD, Runtz EF, Iqbal K, Pavol MJ, 2003 {Role of feedforward
  control of movement stability in reducing slip-related balance loss and falls
  among older adults.}
\newblock \emph{Journal of Neurophysiology} \textbf{90}:755--62.
\newblock \doi{10.1152/jn.01118.2002}.

\bibitem{Fujimoto2013}
Fujimoto M, Chou LS, 2013 {Region of Stability Derived by Center of Mass
  Acceleration Better Identifies Individuals with Difficulty in Sit-to-Stand
  Movement}.
\newblock \emph{Annals of Biomedical Engineering} \textbf{42}:1--9.
\newblock \doi{10.1007/s10439-013-0945-9}.

\bibitem{BogleThorbahn1996}
{Bogle Thorbahn} LD, Newton RA, 1996 {Use of the Berg Balance Test to Predict
  Falls in Elderly Persons}.
\newblock \emph{Physical Therapy} \textbf{76}:576--583.

\bibitem{Steffen2002}
Steffen TM, Hacker TA, Mollinger L, 2002 {Age- and Gender-Related Test
  Performance in Community-Dwelling Elderly People: Six-Minute Walk Test, Berg
  Balance Scale, Timed Up {\&} Go Test, and Gait Speeds}.
\newblock \emph{Physical Therapy} \textbf{82}:128--137.

\bibitem{Muir2008a}
Muir SW, Berg K, Chesworth B, Speechley M, 2008 {Use of the Berg Balance Scale
  for predicting multiple falls in community-dwelling elderly people: a
  prospective study.}
\newblock \emph{Physical Therapy} \textbf{88}:449--59.
\newblock \doi{10.2522/ptj.20070251}.

\bibitem{Shumway-Cook2000}
Shumway-Cook A, Brauer S, Woollacott M, 2000 {Predicting the Probability for
  Falls in Community-Dwelling Older Adults Using the Timed Up {\&} Go Test}.
\newblock \emph{Physical Therapy} \textbf{80}:896--903.

\bibitem{Prajna2004b}
Prajna S, Jadbabaie A, 2004 {Safety Verification of Hybrid Systems Using
  Barrier Certificates}.
\newblock \emph{Hybrid Systems: Computation and Control} pages 477--492.

\bibitem{Topcu2007}
Topcu U, Packard A, Seiler P, 2008 {Local stability analysis using simulations
  and sum-of-squares programming}.
\newblock In \emph{Automatica}, volume~44, pages 2669--2675. IEEE.
\newblock \doi{10.1016/j.automatica.2008.03.010}.

\bibitem{Mitchell2005}
Mitchell IM, Bayen AM, Tomlin CJ, 2005 {A time-dependent Hamilton-Jacobi
  formulation of reachable sets for continuous dynamic games}.
\newblock \emph{IEEE Transactions on Automatic Control} \textbf{50}:947--957.
\newblock \doi{10.1109/TAC.2005.851439}.

\bibitem{Lyapunov1992}
Lyapunov AM, 1992 {The general problem of the stability of motion}.
\newblock \emph{International Journal of Control} \textbf{55}:531--534.
\newblock \doi{10.1080/00207179208934253}.

\bibitem{Parrilo2000}
Parrilo PA, 2000.
\newblock {Structured semidefinite programs and semialgebraic geometry methods
  in robustness and optimization}.

\bibitem{maidens2013lagrangian}
Maidens JN, Kaynama S, Mitchell IM, Oishi MMK, Dumont GA, 2013 {Lagrangian
  methods for approximating the viability kernel in high-dimensional systems}.
\newblock \emph{Automatica} \textbf{49}:2017--2029.
\newblock \doi{doi:10.1016/j.automatica.2013.03.020}.

\bibitem{Henrion2013}
Henrion D, Korda M, August D, 2013 {Convex computation of the region of
  attraction of polynomial control systems}.
\newblock \emph{IEEE Transactions on Automatic Control} \textbf{59}:1--1.
\newblock \doi{10.1109/TAC.2013.2283095}.

\bibitem{Majumdar2014}
Majumdar A, Vasudevan R, Tobenkin MM, Tedrake R, 2014 {Convex optimization of
  nonlinear feedback controllers via occupation measures}.
\newblock \emph{The International Journal of Robotics Research}
  \textbf{33}:1209--1230.
\newblock \doi{10.1177/0278364914528059}.

\bibitem{Shia2014b}
Shia V, Vasudevan R, Bajcsy R, Tedrake R, 2014 {Convex Computation of the
  Reachable Set for Controlled Polynomial Hybrid Systems}.
\newblock In \emph{IEEE Conference on Decision and Control}.
\newblock \doi{DOI: 10.1109/CDC.2014.7039612}.

\bibitem{Mohan2016}
Mohan S, Shia V, Vasudevan R, 2016 {Convex Computation of the Reachable Set for
  Hybrid Systems with Parametric Uncertainty} page~25.

\bibitem{Koopman1995}
Koopman B, Grootenboer HJ, de~Jongh HJ, 1995 {An inverse dynamics model for the
  analysis, reconstruction and prediction of bipedal walking}.
\newblock \emph{Journal of Biomechanics} \textbf{28}:1369--1376.
\newblock \doi{10.1016/0021-9290(94)00185-7}.

\bibitem{Hargraves1987}
Hargraves CR, Paris SW, 1987 {Direct trajectory optimization using nonlinear
  programming and collocation}.
\newblock \emph{Journal of Guidance, Control, and Dynamics}
  \textbf{10}:338--342.
\newblock \doi{10.2514/3.20223}.

\bibitem{Flash1985}
Flash T, Hogan N, 1985 {The coordination of arm movements: an experimentally
  confirmed mathematical model.}
\newblock \emph{The Journal of Neuroscience} \textbf{5}:1688--703.

\bibitem{Uno1989}
Uno Y, Kawato M, Suzuki R, 1989 {Formation and control of optimal trajectory in
  human multijoint arm movement}.
\newblock \emph{Biological Cybernetics} \textbf{61}.
\newblock \doi{10.1007/BF00204593}.

\bibitem{Kawato1999}
Kawato M, 1999 {Internal models for motor control and trajectory planning}.
\newblock \emph{Current Opinion in Neurobiology} \textbf{9}:718--727.
\newblock \doi{10.1016/S0959-4388(99)00028-8}.

\bibitem{Todorov2002}
Todorov E, Jordan MI, 2002 {Optimal feedback control as a theory of motor
  coordination.}
\newblock \emph{Nature Neuroscience} \textbf{5}:1226--35.
\newblock \doi{10.1038/nn963}.

\bibitem{Cusumano2013}
Cusumano JP, Dingwell JB, 2013 {Movement variability near goal equivalent
  manifolds: fluctuations, control, and model-based analysis.}
\newblock \emph{Human movement science} \textbf{32}:899--923.
\newblock \doi{10.1016/j.humov.2013.07.019}.

\bibitem{Liu2007}
Liu D, Todorov E, 2007 {Evidence for the flexible sensorimotor strategies
  predicted by optimal feedback control.}
\newblock \emph{The Journal of Neuroscience} \textbf{27}:9354--68.
\newblock \doi{10.1523/JNEUROSCI.1110-06.2007}.

\bibitem{Dingwell2010}
Dingwell JB, John J, Cusumano JP, 2010 {Do humans optimally exploit redundancy
  to control step variability in walking?}
\newblock \emph{PLoS Computational Biology} \textbf{6}:e1000856.
\newblock \doi{10.1371/journal.pcbi.1000856}.

\bibitem{Dingwell2013}
Dingwell JB, Smallwood RF, Cusumano JP, 2013 {Trial-to-trial dynamics and
  learning in a generalized, redundant reaching task.}
\newblock \emph{Journal of Neurophysiology} \textbf{109}:225--37.
\newblock \doi{10.1152/jn.00951.2011}.

\bibitem{Callier1994}
Callier FM, Desoer CA, 1994 \emph{{Linear System Theory}}.
\newblock Springer New York.

\bibitem{rudin1964principles}
Rudin W, 1964 \emph{Principles of mathematical analysis}, volume~3.
\newblock McGraw-Hill New York.

\bibitem{Full2002}
Full RJ, Kubow T, Schmitt J, Holmes P, Koditschek D, 2002 {Quantifying dynamic
  stability and maneuverability in legged locomotion.}
\newblock \emph{Integrative and comparative biology} \textbf{42}:149--57.
\newblock \doi{10.1093/icb/42.1.149}.

\bibitem{Spagna2007}
Spagna JC, Goldman DI, Lin PC, Koditschek DE, Full RJ, 2007 {Distributed
  mechanical feedback in arthropods and robots simplifies control of rapid
  running on challenging terrain.}
\newblock \emph{Bioinspiration {\&} biomimetics} \textbf{2}:9--18.
\newblock \doi{10.1088/1748-3182/2/1/002}.

\bibitem{Maki1997a}
Maki BE, McIlroy WE, 1997 {The Role of Limb Movements in Maintaining Upright
  Stance: The "Change-in-Support" Strategy}.
\newblock \emph{Physical Therapy} \textbf{77}:488--507.

\bibitem{Wand1980}
Wand P, Prochazka A, Sontag KH, 1980 {Neuromuscular responses to gait
  perturbations in freely moving cats}.
\newblock \emph{Experimental Brain Research} \textbf{38}.
\newblock \doi{10.1007/BF00237937}.

\bibitem{Hof2010}
Hof AL, Vermerris SM, Gjaltema WA, 2010 {Balance responses to lateral
  perturbations in human treadmill walking.}
\newblock \emph{The Journal of Experimental Biology} \textbf{213}:2655--64.
\newblock \doi{10.1242/jeb.042572}.

\bibitem{Hughes1994}
Hughes MA, Weiner DK, Schenkman ML, Long RM, Studenski SA, 1994 {Chair rise
  strategies in the elderly.}
\newblock \emph{Clinical Biomechanics} \textbf{9}:187--92.
\newblock \doi{10.1016/0268-0033(94)90020-5}.

\bibitem{Aissaoui1999}
Aissaoui R, Dansereau J, 1999 {Biomechanical analysis and modelling of sit to
  stand task: a literature review}.
\newblock \emph{IEEE International Conference on Systems, Man, and Cybernetics}
  \textbf{1}:141--146.
\newblock \doi{10.1109/ICSMC.1999.814072}.

\bibitem{Anan2012}
Anan M, Ibara T, Kito N, Shinkoda K, 2012 {The Clarification of the Strategy
  during Sit-to-Stand Motion from the Standpoint of Mechanical Energy
  Transfer}.
\newblock \emph{Journal of Physical Therapy Science} \textbf{24}:231--236.
\newblock \doi{10.1589/jpts.24.231}.

\bibitem{OPT464508}
AMTI.
\newblock {OPT464508}.

\bibitem{PhaseSpace}
PhaseSpace {Impulse X2}.
\newblock Technical report.

\bibitem{Recap2}
PhaseSpace, 2010 {Recap2 User's Guide}.
\newblock Technical report, PhaseSpace.

\bibitem{DPM_model}
{Wolfram Research}.
\newblock {Double Pendulum}.

\bibitem{Pavol2002}
Pavol MJ, Owings TM, Grabiner MD, 2002 {Body segment inertial parameter
  estimation for the general population of older adults}.
\newblock \emph{Journal of Biomechanics} \textbf{35}:707--712.
\newblock \doi{10.1016/S0021-9290(01)00250-0}.

\bibitem{matlab_optimization_toolbox}
{The MathWorks}, 2015 {MATLAB: Optimization Toolbox User Guide}.
\newblock Technical report, The MathWorks, Inc.

\bibitem{spotless}
{SPOTLESS}.

\bibitem{mosek}
{APS Mosek}.
\newblock {The MOSEK optimization software}.

\bibitem{Libby2012}
Libby T, Moore TY, Chang-Siu E, Li D, Cohen DJ, Jusufi A, Full RJ, 2012
  {Tail-assisted pitch control in lizards, robots and dinosaurs.}
\newblock \emph{Nature} \textbf{481}:181--4.
\newblock \doi{10.1038/nature10710}.

\end{thebibliography}

\end{document}